\documentclass[11pt,DIV=12]{scrartcl}

\usepackage{amsthm}
\usepackage{amsmath,amssymb,amsfonts,thmtools,hyperref,graphicx,tikz,algorithm2e,adjustbox,algorithm2e}
\usepackage{stmaryrd,color,graphicx}
\usepackage{here,caption,subcaption}
\usepackage{microtype}
\usepackage{comment}
\usepackage[]{forest}\usetikzlibrary{decorations,decorations.pathmorphing,decorations.pathreplacing}

\usepackage[textwidth=2.5cm,disable]{todonotes}
\usepackage[appendix=append]{apxproof} \usepackage[capitalize,noabbrev]{cleveref}

\newcommand{\cut}[1]{}

\newcommand{\smid}{\mathbin{;}}
\newcommand{\qr}{\textrm{qr}}
\renewcommand{\phi}{\varphi}
\renewcommand{\epsilon}{\varepsilon}
\newcommand{\onot}{\mathcal O}

\newcommand{\false}{\texttt{false}}
\newcommand{\model}[4]{\llbracket#2(#3\smid#4)\rrbracket^#1}
\newcommand{\tree}{T}
\newcommand{\train}{S}
\newcommand{\bool}{\ensuremath{\{+,-\}}}
\newcommand{\mso}{\text{MSO}}
\newcommand{\labelChanges}{U}
\newcommand{\order}{\leq}
\renewcommand{\theta}{\vartheta}
\newcommand{\ie}{\mbox{i.e. }}

\newcommand{\NN}{\mathbb{N}}

\DeclareMathOperator*{\pr}{\text{Pr}}

\newcommand{\braces}  [1] {{\left\lbrace #1 \right\rbrace}}

\newcommand{\convexpath}[2]{
[   
    create hullnodes/.code={
        \global\edef\namelist{#1}
        \foreach [count=\counter] \nodename in \namelist {
            \global\edef\numberofnodes{\counter}
            \node at (\nodename) [draw=none,name=hullnode\counter] {};
        }
        \node at (hullnode\numberofnodes) [name=hullnode0,draw=none] {};
        \pgfmathtruncatemacro\lastnumber{\numberofnodes+1}
        \node at (hullnode1) [name=hullnode\lastnumber,draw=none] {};
    },
    create hullnodes
]
($(hullnode1)!#2!-90:(hullnode0)$)
\foreach [
    evaluate=\currentnode as \previousnode using \currentnode-1,
    evaluate=\currentnode as \nextnode using \currentnode+1
    ] \currentnode in {1,...,\numberofnodes} {
-- ($(hullnode\currentnode)!#2!-90:(hullnode\previousnode)$)
  let \p1 = ($(hullnode\currentnode)!#2!-90:(hullnode\previousnode) - (hullnode\currentnode)$),
    \n1 = {atan2(\y1,\x1)},
    \p2 = ($(hullnode\currentnode)!#2!90:(hullnode\nextnode) - (hullnode\currentnode)$),
    \n2 = {atan2(\y2,\x2)},
    \n{delta} = {-Mod(\n1-\n2,360)}
  in 
    {arc [start angle=\n1, delta angle=\n{delta}, radius=#2]}
}
-- cycle
}

\newtheorem{theorem}{Theorem}[section]
\newtheorem{corollary}[theorem]{Corollary}
\theoremstyle{plain}
\newtheoremrep{lemma}[theorem]{Lemma}
\newtheoremrep{conjecture}[theorem]{Conjecture}
\newtheoremrep{observation}[theorem]{Observation}

\newtheorem{example}[theorem]{Example}
\newtheorem{claim}[theorem]{Claim}

\crefname{algocf}{Algorithm}{Algorithms}

\begin{document}

\title{Learning definable hypotheses on trees}
\author{\large Emilie Grienenberger\\\normalsize ENS Paris-Saclay \\\normalsize emilie.grienenberger@ens-cachan.fr
\and \large Martin Ritzert\\\normalsize RWTH Aachen
  University\\
  \normalsize ritzert@informatik.rwth-aachen.de}
\date{}

\maketitle

\begin{abstract}
  We study the problem of learning properties of nodes in tree structures.
  Those properties are specified by logical formulas, such as formulas from first-order or monadic second-order logic.
  We think of the tree as a database encoding a large dataset and therefore aim for learning algorithms which depend at most sublinearly on the size of the tree. 
  We present a learning algorithm for quantifier-free formulas where the running time only depends polynomially on the number of training examples, but not on the size of the background structure. 
  By a previous result on strings we know that for general first-order or monadic second-order (MSO) formulas a sublinear running time cannot be achieved.
  However, we show that by building an index on the tree in a linear time preprocessing phase, we can achieve a learning algorithm for MSO formulas with a logarithmic learning phase.
\end{abstract}

\section{Introduction}
In this paper we study the algorithmic complexity of learning properties of nodes in directed labeled trees using a declarative framework introduced by Grohe and Turán \cite{grohe2004learnability}.
Let $\tree$ be such a tree with nodes $V(\tree)$.
We call $\tree$ the \emph{background tree} of our learning problem.
The tree $\tree$ encodes the background knowledge of the learning problem and thus provides the information on which the classification of the nodes $u\in V(\tree)$ can be based.
In our setting a (boolean) \emph{classifier} is a function $H\colon V(\tree)\rightarrow \bool$ that  estimates whether a given node admits a certain property.
A learning algorithm gets a \emph{training set} $\train\subseteq V(\tree)\times \bool$, that is a set of pairs $(u,c)$ of positive and negative \emph{examples}, and the background tree $\tree$ as input and returns a classifier $H_S\colon V(\tree)\rightarrow \bool$ as its \emph{hypothesis}. 
We say that learning was successful if the hypothesis $H_S$ is \emph{consistent} with $\train$ which means that for every $(u,c)\in \train$ we have that $H_S(u)=c$.
Achieving consistency with $\train$ can be seen as the extreme case of minimizing the \emph{training error}, \ie the number of $(u,c)\in \train$ such that $H_S(u)\neq c$.
Minimizing the training error, also called \emph{empirical risk minimization}, results in provably good generalization behavior in the PAC learning model (see \cite{bluehrhau+89}).
For those generalization results, we use that logical formulas on trees admit bounded VC-dimension (shown by Grohe and Turán \cite{grohe2004learnability}) and that any consistent learner can be turned into a PAC learner by an appropriate training set $S$ (see \cite{bluehrhau+89}).
We give more details on the connection to PAC learning in Section \ref{subsec:pacLearning}.

An example of a simple property a node can admit is having an ancestor with label $b$.
This property can be expressed by the logical formula $\phi(x) = \exists y\, (y<x) \land R_b(y)$.
We aim to learn properties which can be defined by logical formulas with parameters based on positive and negative examples over tree-structured data such as web pages, XML databases and JSON files.

\begin{example}\label{example:adDetection}
Given a large website such as a news portal.
The document object model (DOM) of a website is the tree of elements which form the website.
Many websites contain a number of ads and some of them are not trivially detectable.
A learning algorithm could then estimate the property of a position in the DOM to be part of some ad.

The output of a learning algorithm would then be a formula that distinguishes nodes belonging to the content of the web page from those belonging to ads.
Those formulas could then be used as a basis for new simpler or better filter rules.
\end{example}

In Example \ref{example:adDetection}, the user could select parts of a web page, which he sees as advertisement and then let the learning algorithm produce a classifier which is consistent with his choice.
In this paper we will not go into detail of how we get our training set but instead only talk about finding consistent hypotheses for a given training set.

We consider learning algorithms that return classifiers based on logical formulas, especially quantifier-free formulas and formulas from monadic second-order logic.
In our logical framework, a classifier consists of a formula $\phi(x\smid \bar y)$ and an instantiation $\bar v$ of the free variables $\bar y$.
This formula $\phi$ has two types of free variables; we refer to $x$ as the \emph{instance variable} and to $\bar y = (y_1,\dots, y_\ell)$ as \emph{parameter variables}, where $\ell\in\mathbb{N}$. 
The background structure $\tree$ is the (fixed) background knowledge that encodes the context of a node which is to be classified.
The parameters of $\phi$, which can be seen as constants the formula is allowed to use for the classification, are taken from $V(\tree)$.
In Example \ref{example:adDetection}, a classifier would consist of a logical formula and a number of positions in the DOM of the web page.
The formula $\phi(x\smid \bar y)$ together with an $\ell$-tuple $\bar v \in V(\tree)^\ell$ of parameters then defines a binary classifier $\model{\tree}{\phi}{x}{\bar v}\colon V(\tree) \rightarrow \bool$ over the tree $\tree$ as follows.
An instance $u\in V(\tree)$ is classified as positive if $\tree \vDash \phi(u,\bar v)$ such that we have $\model{\tree}{\phi}{x}{\bar v}(u) = +$.
Correspondingly we have $\model{\tree}{\phi}{x}{\bar v}(u') = -$ for any $u'\in V(\tree)$ with $\tree \nvDash \phi(u',\bar v)$.
We call such a classifier where $\phi$ is an MSO formula an \emph{MSO definable hypothesis}.

We assume the background tree $\tree$ to be very large, which means large enough that just reading it sequentially takes long, while the logical formula returned by the learning algorithm is assumed to be small for every real-world query.
This implies two things.
First, we use a \emph{data complexity} view for the analysis considering the tree $\tree$ and the training set $\train$ as data and the hypothesis class parameterized by $\ell$ (and an additional parameter $q$ introduced later) as a constant, such that the complexity results are only given in terms of $\tree$ and $\train$.
Essentially this means that the influence of the formula is considered to be constant. Second, we are interested in finding algorithms which run in sublinear time in the size of~$\tree$.
As such sublinear algorithms are unable to read the whole background tree $\tree$, we model the exploration of $\tree$ using oracles.
Those oracles allow the learning algorithm to explore the tree by following edges, starting from the training examples in $\train$.
This is formally defined in Section \ref{sec:preliminaries}.

In general there are no consistent sublinear learning algorithms for first-order and monadic second-order formulas over trees.
In \cite{DBLP:conf/alt/GroheLR17} the authors have shown that for learning first-order formulas over words, linear time is necessary.
The same counterexample can also be used for trees, showing that linear time is again necessary.
For structures of bounded degree, there exists a sublinear learning algorithm for first-order formulas (see \cite{grorit17}).
This result is not applicable in our setting, as the ancestor relation $\leq$ in the signature of our trees induces unbounded degree (the degree of the root is $|V(\tree)|-1$ as it is an ancestor of every other node).

\subsection{Our Results}
In our formal setting, we consider learnability on trees for monadic second-order logic and the quantifier-free fragment of first-order logic.
We show that in contrast to the corresponding case on strings (see \cite{DBLP:conf/alt/GroheLR17}), even the relatively simple task of learning quantifier-free formulas on trees needs at least linear time.
This is due to the need to synthesize appropriate parameters, a task which can involve searching for the largest common ancestor of two nodes.
We show that this is really the core of the problem by giving a sublinear learning algorithm for quantifier-free formulas in Section \ref{sec:quantifier_free} where we exploit an additional oracle providing access to the largest common ancestor of two nodes.

As we know that there is no sublinear learning algorithm for MSO formulas on trees, we investigate whether the necessary linear computation depends on the training examples.
This hardness still holds if the learning algorithm is allowed to return more parameters, as long as the complete training set can not be encoded in those parameters.
It turns out that it is possible to build an auxiliary structure in linear time which can then be used for sublinear learning.
We present an algorithm which builds such an index structure without knowledge of the training set in linear time and then uses logarithmic time to output a consistent hypothesis.
The algorithm builds on the results on strings (see \cite{DBLP:conf/alt/GroheLR17}) as well as known techniques for evaluating MSO formulas under updates (see \cite{BalminPV04}), combining them in a non-trivial way to a learning algorithm for MSO formulas.
The linear indexing phase in the learning algorithm cannot be avoided since already for first-order formulas on words it is necessary to invest at least linear time in $\onot(|\tree|)$ to find a consistent hypothesis.
The following theorem is the main result of this paper. 
 
{
\begin{theorem}\label{thm:main}
There is a consistent MSO learning algorithm on trees which uses linear indexing time $\onot(|\tree|)$ and logarithmic learning time $\onot(|\train| \log |\tree|)$.
\end{theorem}
}

As an application of our indexing algorithm we describe an online learning algorithm that computes an index in $\onot(|\tree|)$ and then, for a sequence of examples, updates its MSO-definable hypothesis in time $\onot{}(\log^2(\tree))$ per example.
That is, in this setting the examples arrive one-by-one and we are able to maintain a consistent hypothesis in polylogarithmic time in the background structure.

\subsection{Related work}
The field of inductive logic programming  (see for example \cite{cohpag95,kiedze94,mug91,mug92,mugder94}) is very close to our framework.
In both cases the aim is to infer logical formulas from positive and negative examples such that the logical formula is consistent with the training examples.
The main difference to our setting is that in the ILP framework the background knowledge is also encoded in logic (a so called background theory), while we use a structure to encode background knowledge.
In our setting, facts such as gender and age of a person or the issuing institute of a credit card are represented using nodes for person and credit card, as well as unary relations to describe their attributes.
Naturally facts which involve multiple entities can be represented by edges.
Our framework is able to represent such facts as long as the union of all binary relations still describes a tree or forest while in the ILP setting there is no such restriction. The other important difference is that ILP focuses on first-order logic (and there especially Horn-formulas), while in this paper we work with monadic second-order logic (MSO) which is strictly more expressive than first-order logic.
There is a number of other logical frameworks for machine learning, mainly originating from the field of formal verification and databases. Examples are given by \cite{aboangpap+13,bonciusta16,lodmadnei16,garneimadrot16,jorkai16,weiss2017reverse}.

Another related field, which is based on the query by example strategy, is to learn XPATH queries as in \cite{staworko2012learning}.
There, unary relations defined by an XPATH expression are learned for arbitrary training sets. 
The main difference to our setting is that we use MSO formulas, which are in general more expressive than XPATH statements and then restrict the maximal complexity of our formulas.

The field of automata learning and learning of regular languages is also to some degree similar to our setting,  especially since we are also applying automata based techniques.
There are numerous negative results such as \cite{Angluin78,Gold78,PittW93,KearnsValiant94,Angluin90}.
Of the positive results in that area \cite{Angluin87,RivestS93,GarciaO92,drewes2003learning}, most of them use an active framework where a teacher iteratively gives counterexamples until the hypothesis is correct.
In our framework a consistent hypothesis for a training set is sought and that training set is known from the beginning.
Even though our classification problem can be encoded as a learning problem for regular tree languages, their results seem technically unrelated to ours.

\section{Preliminaries}\label{sec:preliminaries}
In this paper we work with logical formulas from the quantifier-free fragment of first-order logic and formulas from monadic second-order logic.
Quantifier-free formulas only consist of boolean combinations of atomic properties, while monadic second-order logic (MSO) extends first-order logic (FO) by quantification over sets of nodes.
As an example, take the MSO formula $\exists X \forall z\, Xz$ which is always satisfied as there is always a set $X$ containing all elements of the structure. 
It is known that a set of trees can be recognized by a deterministic bottom-up tree automaton (DTA) $\mathcal A$ if and only if it can be characterized by an MSO sentence $\Phi$ and both $\mathcal A$ and $\Phi$ can be computed from each other.
For a more detailed description of MSO and tree automata we refer to \cite{tho97a}.

In this paper we consider labeled trees as background structures.
The most prominent examples for trees in a database context are the tree-structured data exchange formats XML and JSON.
Formally, a labeled tree $\tree = (V(\tree), E_1,E_2,R_1,\dots,R_r,\leq)$ is a structure with vertex set $V(\tree)$ and binary edge relations $E_1$ and $E_2$ encoding the first and second child of a node in a binary tree, or in case of unranked trees, the first child and the next sibling of each node.
The unary relations $R_1,\dots R_r$ define the label of each node and for every $a,b\in V(\tree)$ we have $a\leq b$ if $a$ is an ancestor of $b$.
In this paper we use formulas over the alphabet $\sigma = \{ E_1,E_2,\leq,R_1,\dots, R_r\}$ which means that in a formula we can access the tree structure using $E_1$ and $E_2$ as well as the labels of each node using $R_1,\dots R_r$.
A fragment of an XML document (focusing on persons) could look like the following.
\begin{verbatim}
<person name="A. Turing" birthday="1912-06-23">
  <interest>computer science</interest>
  <interest>marathon running</interest>
  ...
</person>
\end{verbatim}
A formula has access to all tags occurring in the XML document.
In the above XML fragment, the labels are given by the unary relations \(\texttt{`person',`name',`birthday',`interest'}\).
Adding additional content-based labels such as $\texttt{`computer scientist'}$ or $\texttt{`runner'}$ to the set of unary relations allows to write formulas that depend on the content of nodes.
Without such content-based labels 
a formula could only use structural properties and define sets such as all persons with at least three interests and two friends.

\subsection{Learning Model}\label{subsec:learning}
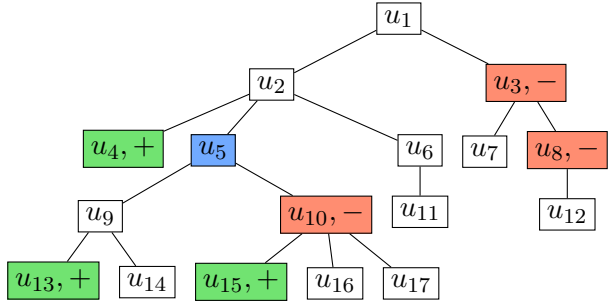
\begin{figure}[t]
  \begin{center}
    \vspace{-1ex}
    \colorlet{parameterblue}{rgb:cyan,2;white,4;blue,3}
\colorlet{examplegreen}{rgb:green,4;white,4;black,1}
\colorlet{examplered}{rgb:orange,2;red,3;white,4}
\begin{forest}
  for tree={l sep=1em, l=1em, math content, inner sep=0.25em,draw,},
  [u_1
    [u_2
      [{u_4,+},fill=examplegreen]
      [u_5,fill=parameterblue
        [u_9
          [{u_{13},+},fill=examplegreen]
          [u_{14}]
        ]
        [{u_{10},-},fill=examplered
          [{u_{15},+},fill=examplegreen]
          [u_{16}]
          [u_{17}]
        ]
      ]
      [u_6[u_{11}]]
    ]
    [{u_3,-},fill=examplered
      [u_7]
      [{u_8,-},fill=examplered[u_{12}]]
    ]
  ]
\end{forest}
     \vspace{-1ex}
  \end{center}
  \caption{A tree with positive (green) and negative (red) example nodes}\label{fig:definableSets}
\end{figure}

We consider the model of \emph{supervised learning} where the input to a \emph{learning algorithm} is a \emph{training set} (or training sequence) $\train\subseteq V(\tree)\times\bool$. Each \emph{example} $(u,c)\in\train$ consists of a node $u$ and its classification $c$.
We write $u\in\train$ when we are not interested in $u$'s classification but only the position of the example.
We assume that $\train$ is non-contradicting, that is if $(u,c)\in\train$, then $(u,\neg c)\notin\train$.
For the tree given in Figure \ref{fig:definableSets} we define the training set: 
\[
  \train=\{(u_3,-),(u_4,+),(u_8,-),(u_{10},-),(u_{13},+),(u_{15},+)\}
\]
As already defined in the introduction, a definable hypothesis $\model{\tree}{\phi}{x}{\bar v}$ assigns $+$ to every position $u\in V(\tree)$ with $\tree \vDash \phi(u,\bar v)$ and $-$ otherwise.
Let $\phi(x\smid y) = \exists z (E(x,z) \land E(z,y)) \land x\neq y $ accepting all positions with a distance of $2$ from the position of $y$ ($\tree$ contains no self-loops). In the example from Figure \ref{fig:definableSets} we have
$\model{\tree}{\phi}{x}{u_5}(u) = +$ if and only if $u \in {\{u_1,u_4,u_6,u_{13},u_{14},u_{15},u_{16},u_{17}\}}$.
This hypothesis is \textit{consistent} with $\train$ as it accepts all positive and none of the negative examples from $\train$.
The formula $\psi(x) = \exists\, z E_1(z,x)$ defines another consistent hypothesis with $\llbracket \psi(x) \rrbracket ^\tree (u) = + \text{ if and only if } u \in \{ u_2,u_4,u_7,u_9,u_{11},u_{12},u_{13},u_{15}\}$.

The quantifier rank $\qr(\phi)$ of a formula $\phi$ is the maximal nesting depth of quantifiers in~$\phi$. As every set $S^+\subseteq V(\tree)$ is definable by a (long enough) MSO formula $\phi$, we restrict our study to sets which are definable by formulas $\phi(x\smid y_1,\dots y_\ell)$ with $\qr(\phi)\leq q$ where $q$ and $\ell$ are considered to be part of the problem.
Restricting $q$ and $\ell$ reduces the risk of overfitting as such restricted formulas can only memorize a bounded number of positions and thus, on larger training sets, have to exploit more general patterns in the data.

Our framework naturally admits two different learning problems.
In \emph{model learning} we assume that there is a consistent classifier $\model{\tree}{\phi}{x}{\bar v}$ with $\bar v \in V(\tree)^\ell$ and $\qr(\phi)\leq q$, but only $q$ and $\ell$ are given to the learning algorithm.
This reflects the assumption that there is a simple, as expressed by the choice of $q$ and $\ell$, but unknown pattern behind the classification of the training examples from $\train$.
In \emph{parameter learning} the formula $\phi$ of a consistent classifier $\model{\tree}{\phi}{x}{\bar v}$ is fixed. 
The learning algorithm is not allowed to modify $\phi$ and has to find a consistent parameter setting $\bar v\in V(\tree)^\ell$.
This variant reflects the case where we have a general idea about how the solution looks like, but are missing the details.
Counterintuitively, parameter learning is the harder problem: the restriction to a specific formula $\phi$ might impose (unnecessary) restrictions on the parameters.
An edifying example is the parameter learning problem on a background structure $\tree$ with a singleton unary relation $R$, using the formula $\vartheta(x\smid y) = R(y)$ and a training set $\train=\{(u,+)\}$ for an arbitrary $u\in V(\tree)$.
In this example, for any fixed traversal strategy of the learning algorithm on $V(\tree)$, the single possible parameter can be placed in the position which is evaluated last.
For the associated model learning problem, a possible solution would be to return the formula $\psi(x) = \texttt{true}$ without parameters, which defines a hypothesis consistent with $\train$.
In the example given in Figure \ref{fig:definableSets}, a learning algorithm for the model learning problem would be free to choose between any consistent hypothesis, such as the example formulas $\phi(x\smid u_5)$ and $\psi(x)$ given above.
In the parameter learning problem with the formula $\phi(x\smid y)$, the (only) consistent output is the assignment $y=u_5$.

In practice, whenever we want to evaluate a definable hypothesis $\model{\tree}{\phi}{x}{\bar v}$, we have to solve an instance of the model checking problem for the formula $\phi$ over the structure $\tree$.
For the case of a fixed MSO formula $\phi(x\smid \bar y)$ on a tree $\tree$, there is an evaluation strategy in $\onot(|\tree|)$ using tree automata, see for example \cite{tho97a}, while a quantifier-free formula $\theta$ can be evaluated in time $\onot(|\theta|)$.

\subsection{Access Model}\label{subsec:access_model}

A sublinear learning algorithm is unable to read the whole background structure during its computation. We therefore model the access to the background structure by oracles.

\begin{description}
	\item[Relation Oracles] For a $k$-ary relation $R$, the corresponding oracle returns on input of \\\(\bar u \in V(\tree)^k\) whether $\bar u \in R$ holds in $\tree$. 
	\item[Neighborhood Oracle] On input of a node $u\in V(\tree)$, returns the $1$-neighborhood of $u$ in $\tree$.
\end{description}

In a binary tree, the $1$-neighborhood $N$ of a node $u$ consists of $u$, its parent and its child nodes.
In the unranked case $N(u)$ consists of $u$ as well as its left and right sibling, first child and parent.
This access model is called \emph{local access} as the tree can only be explored by following edges. Directly jumping to the closest node that is in a relation $R$ or to the last child of an unranked node is not possible in this access model. 
In practice, those oracles can be implemented using random access on the background tree $\tree$. \begin{figure}[t]
  \begin{center}
    \vspace{-1ex}
    \begin{adjustbox}{valign=b}
\begin{forest}
   for tree={draw,circle,l=0.7em,l sep=0.7em,},
   [,phantom
    [,phantom]
    [,phantom
      [,
        [,no edge,color=white][,phantom]
      ]
      [,phantom]
    ]
   ]
\end{forest}
\end{adjustbox}\qquad~~
\begin{adjustbox}{valign=b}
\begin{forest}
   for tree={draw,circle,l=0.7em,l sep=0.7em,},
   [,phantom
    [,phantom]
    [
      [,fill=blue!50
        [][]
      ]
      [,phantom]
    ]
   ]
\end{forest}
\end{adjustbox}\qquad~~
\begin{adjustbox}{valign=b}
\begin{forest}
   for tree={draw,circle,l=0.7em,l sep=0.7em,},
   [
    [,phantom]
    [,fill=blue!50
      [,fill=blue!20
        [][]
      ]
      []
    ]
   ]
\end{forest}
\end{adjustbox}\qquad~~
\begin{adjustbox}{valign=b}
\begin{forest}
   for tree={draw,circle,l=0.7em,l sep=0.7em,},
   [,fill=blue!50
    []
    [,fill=blue!20
      [,fill=blue!20
        [][]
      ]
      []
    ]
   ]
\end{forest}
\end{adjustbox}
\vspace{-1em}   \end{center}
  \caption{From left to right: Exploration of a tree using neighborhood queries starting from an example node. The dark blue node is the one on which the neighborhood query has been performed last}\label{fig:localAccess}
\end{figure}
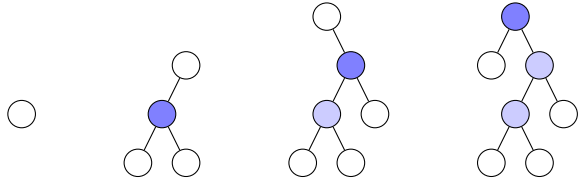
A learning algorithm using local access starts with all nodes occurring in the training set $\train$ and then explores the background tree $\tree$ using the neighborhood and relation oracles.
Figure \ref{fig:localAccess} illustrates how the learning algorithm can explore the background tree using subsequent neighborhood queries on the topmost node. 
The neighborhood queries return the vertices, while the relation queries clarify directions and labels.

\subsection{PAC learning}\label{subsec:pacLearning}
It is known that under certain simplicity restrictions, a consistent learner generalizes well to new and potentially unseen examples.
For a \emph{probably approximately correct (PAC)} learning algorithm we have that for every $\epsilon$ and $\delta$ there is a size $s\in \mathbb N$ of the training set $\train$ such that the error of the hypothesis under new examples is bounded by $\epsilon$ with a confidence level of $1-\delta$.
This is made formal in Equation \eqref{eqn:PacCriterion}.
Let $c^*\colon V(\tree) \rightarrow \bool$ be the function that assigns the correct classification to every node $u\in V(\tree)$.
Let the training set $\train \in 2^{V(\tree)\times\bool}$ be a set of $t$ examples $(u, c^*(u))$ chosen independently and identically distributed (i.i.d.) according to a fixed distribution $D$.
Let $(u, c^*(u))\sim D$ and $\train\sim D$ denote the random choices according to $D$.
Let $H_\train\colon V(\tree) \rightarrow \bool$ be the hypothesis returned by the learning algorithm on input of the training set $\train$.
Then the PAC criterion is given by
\begin{align}
    \pr_{\train\sim D} \left(\pr_{(u,c)\sim D} \left(H_\train(u) \neq c^*(u)\right) \leq \epsilon\right) \geq 1-\delta \label{eqn:PacCriterion}
\end{align}
where the outer probability (the \emph{confidence}) $\pr_{\train\sim D}$ is taken over the training set $\train$ for which the learning algorithm produces a hypothesis $H_\train$.
The inner probability is the expected error of the hypothesis $H_\train$ for an example chosen according to the distribution~$D$.
For more details on the PAC learning model, we refer to \cite{valiant1984theory}. 

There is a very general result from learning theory that shows that for any hypothesis class with bounded Vapnik-Chervonenkis (VC) dimension a consistent learner can be turned into a PAC learning algorithm by providing a large enough training set (see \cite{vapnik2015uniform} or \cite{bluehrhau+89}).
For the case of MSO definable hypotheses on trees, the VC dimension is bounded (see \cite{grohe2004learnability}).
Hence, there is a sufficient size $s$ of $\train$, depending polynomially on $\frac{1}{\epsilon}$ and $\frac{1}{\delta}$, to satisfy the PAC criterion.
Using the algorithms from our Theorems \ref{th:sub_qf_binary} and \ref{thm:msoLearning}, each providing a consistent learning algorithm for learning formulas on trees, we have the following corollary.
\begin{corollary}
There are (efficient) PAC learning algorithms for of learning quantifier-free formulas and monadic second-order formulas over trees.
\end{corollary}
The running time of those PAC learning algorithms follows directly from the corresponding theorems.
A PAC learning algorithm is called \emph{efficient} if the dependence of the running time on the number of training examples is polynomial which is the case in Theorem \ref{th:sub_qf_binary} and \ref{thm:msoLearning}.

\newcommand{\suff}{N_\train^\text{suff}} \newcommand{\examplesAbove}{S_\text{above}(v)}
\newcommand{\examplesBelow}{S_\text{below}(v)}
\newcommand{\examplesFamily}{S_\text{family}(v)}
\newcommand{\examplesLeft}{S_\text{left}(v)}
\newcommand{\examplesRight}{S_\text{right}(v)}

\section{Quantifier-free formulas}
\label{sec:quantifier_free}
The class of quantifier-free formulas $\Phi_\sigma[n]$ with $n$ free variables over the signature $\sigma$ is defined as the fragment of first-order logic without quantifiers.
Every formula $\phi(\bar x) \in \Phi_\sigma[n]$ can only compare its free variables directly using a boolean combination of atomic properties from $\sigma$.
We exploit this limitation to obtain a learning algorithm for quantifier-free formulas which runs in constant time with respect to $\tree$ under a slightly relaxed notion of local access.
We start by showing that there is no sublinear learning algorithm using the notion of local access as defined in Section \ref{sec:preliminaries}. 

\begin{lemma}\label{thm:no_subl_without_commonanc}
    For every $\ell\in\mathbb N$, there is no consistent learning algorithm that uses local access and, given a binary tree $\tree$ and a training set $\train$, returns a consistent hypothesis $\model{\tree}{\phi}{x}{\bar v}$ with $\phi\in \Phi_\sigma[\ell+1]$ in time $o(|\tree|)$.
\end{lemma}

\begin{inlineproof} 
    For the contradiction assume that $L$ is such a sublinear learning algorithm. 
    Consider the family of trees $(\tree_m)_{m\in\mathbb{N}}$ shown in Figure \ref{fig:ce_sublinear_local} with $\ell$ fixed to the number of parameters used in the formulas returned by $L$.
    \begin{figure}[t]
      \vspace{-3ex}
      \begin{center}
         \pgfkeys{/forest,
    tria/.style={
      node format={
        \noexpand\node [
          draw,
          shape=regular polygon,
          regular polygon sides=3,
          inner sep=0pt,
          outer sep=0pt,
          \forestoption{node options},
          anchor=\forestoption{anchor}
         ]
        (\forestoption{name}) {\foresteoption{content format}};
      },
      child anchor=north,
      edge path={
        \noexpand\path[\forestoption{edge}]
          (!u.parent anchor) --
          (.child anchor)\forestoption{edge label};
      },
    },
  }

\begin{forest}
  for tree={inner sep=.1em,l sep=0
  },
    [$u$,edge label={node[midway, fill=white,inner sep=.2em,rotate=90]{$\dots$}},l=3em,
      [$v$,l=1.8em,s sep=0.7em,
        [$\overset{u_1}{+}$,edge label={node[midway, fill=white,inner sep=.1em,rotate=43]{$\dots$}},l=3.5em,tier=plus,name=pluslinks]
        [$\overset{u_2}{+}$,edge label={node[midway, fill=white,inner sep=.1em,rotate=60]{$\dots$}},tier=plus]
        [$\overset{u_3}{+}$,edge label={node[midway, fill=white,inner sep=.1em,rotate=83]{$\dots$}},tier=plus]
        [$\cdots$,no edge,tier=plus]
        [$\overset{u_{l+1}}{+}$,edge label={node[midway, fill=white,inner sep=.1em,rotate=-42]{$\dots$}},tier=plus,name=plusrechts]
      ]
      [$v'$,l=1.8em,name=middleright,s sep=0.1em,
          [$\overset{u_{l+2}}{-}$,edge label={node[midway, fill=white,inner sep=.1em,rotate=43]{$\dots$}},l=4em,l sep=2em,tier=plus,name=minuslinks]
          [$\overset{u_{l+3}}{-}$,edge label={node[midway, fill=white,inner sep=.1em,rotate=65]{$\dots$}},tier=plus]
          [$\overset{u_{l+4}}{-}$,edge label={node[midway, fill=white,inner sep=.1em,rotate=87]{$\dots$}},tier=plus]
          [$\cdots$,no edge,tier=plus]
          [$\overset{u_{2l+2}}{-}$,edge label={node[midway, fill=white,inner sep=.1em,rotate=-43]{$\dots$}} ,tier=plus,name=minusrechts]
      ]
    ]\draw[decorate, decoration={brace, amplitude=.5em}]
  ([xshift=.5em]middleright.east) --
  node[xshift=1.2em,yshift=.6em]{$m$}
  ([yshift=0.2em]minusrechts.north);
\draw[decorate, decoration={brace, mirror,amplitude=.5em}]
  (pluslinks.south west) --
  node[below=.5em]{$\ell+1$}
  (plusrechts.south east); 
\draw[decorate, decoration={brace, mirror,amplitude=.5em}]
  (minuslinks.south west) --
  node[below=.5em]{$\ell+1$}
  (minusrechts.south east); 
\end{forest}
       \end{center}
      \vspace{-3ex}
      \caption{Sketch of the tree $\tree_m$ parameterized by $m$ and $\ell$ from \cref{thm:no_subl_without_commonanc}. Substituting $v$ and $v'$ by balanced binary trees gives the binary $\tree_m$ used in \cref{thm:no_subl_without_commonanc}.
      }
      \label{fig:ce_sublinear_local}
    \end{figure}
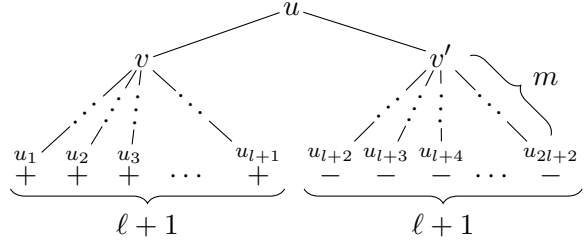
    We define the training set $\train$ of $2\ell+2$ examples as {\(\train = \{(u_1,+), \dots (u_{\ell+1},+), (u_{\ell+2},-),\dots, (u_{2\ell+2},-)\}\)}. 
    In Figure \ref{fig:ce_sublinear_local}, the positions of the nodes $u_i$ are indicated by their classifications $+$ or $-$.
    The size of $\tree_m$ is linear in $m$ for any fixed $\ell$ and therefore we can choose $m$ such that the running time of $L$ on $\tree_m$ is smaller than $m$.
    This is possible since $L$ runs in time $o(|\tree_m|)$.

    The formula $\phi(x\smid y) = y < x$ with parameter $v$ as indicated in \cref{fig:ce_sublinear_local} is consistent with~$\train$.
    Let $\psi(x\smid \bar y)\in \Phi_\sigma[\ell+1]$ be the formula and $v_1,\dots,v_\ell \in V(\tree_m)$ the parameters returned by $L$ on input $\tree_m$ and $\train$.
    Then each $v_i$ is ancestor to exactly one $u \in \train$ as $L$ may only return nodes $v_1,\dots,v_\ell$ it has seen during its computation.
    All example nodes $x_i\in \train$ are leaves of $\tree_m$.
    By choosing the runtime of $L$ on $\tree_m$ to be strictly less than $m$ while the paths adjacent to each leaf have length $m$, $L$ only sees nodes from those paths such that only those can be returned.

    Since $\psi$ is a quantifier-free formula and thus can only directly compare free variables, every parameter $v_i\in V(\tree)$ has the same effect on every example except (possibly) the one below $v_i$.
    There are $\ell+1$ positive and $\ell+1$ negative examples but only $\ell$ parameter, thus at least one example will be misclassified.
    Intuitively this holds as all paths locally look identical and the learning algorithm $L$ is unable to detect on which side of the tree $\tree_m$ an example $u\in\train$ is located.
\end{inlineproof}

In the following we extend the notion of local access to cope with the globality of the ancestor relation.

\subsection{Extended local access}\label{subsec:extended_la}

Our next result states that synthesizing the parameters for quantifier-free formulas is really the core of the linear complexity.
\emph{Extended local access} extends local access by a \emph{common ancestor oracle}: 
\begin{description}
    \item[Common Ancestor Oracle] On input of two nodes $u$ and $v$, it returns
        the lowest node $w$ such that $w\leq u$ and $w\leq v$ (their \emph{lowest common ancestor})
\end{description}

For the example from \cref{thm:no_subl_without_commonanc} a consistent parameter $v$ can be found in constant time using extended local access.
This is generalized in the following theorem.
We show that with the common ancestor oracle quantifier-free formulas are learnable in sublinear time, while without this oracle linear time is necessary.
Let $\ell\in\NN$ be an integer.
For binary trees we use the signature $\sigma=\{E_1,E_2,\leq\}$ with relations for the first and second child as well as the ancestor relation $\leq$.
For unranked trees we use the signature $\sigma=\{E_1,E_2,\leq,\preceq\}$ with edge relations for the first child ($E_1$) and the next sibling ($E_2$).
The relation $\preceq$ is an order over the children of each node, defined as the reflexive transitive closure of $E_2$.

\begin{theorem}\label{th:sub_qf_binary}
    There is a learning algorithm that, on input of a tree $\tree$ and a training set $\train$, outputs a consistent hypothesis $\model{\tree}{\phi}{x}{\bar v}$, where $\phi\in\Phi_\sigma[1+\ell]$ and $\bar v\in V(\tree)^{\ell}$ for binary trees and $\phi\in\Phi_\sigma[1+2\ell]$ ($\bar v \in V(\tree)^{2\ell}$) for unranked trees.
    This algorithm uses extended local access and runs in time $\onot\left(|\train|^2+|\train|^{\ell}|\train|\right)$.
\end{theorem}

\begin{inlineproof}[Proof sketch.]
  Let $\tree$ be a tree and $\train$ a training set.
  A set of nodes $\suff\subseteq V(\tree)$ is \emph{$k$-sufficient} for $\train$ if for every consistent hypothesis based on $\phi\in\Phi_\sigma[1+\ell]$, there is a consistent hypothesis based on $\phi'\in\Phi_\sigma[1+k\ell]$ using only parameters from $\suff$. 
  Note that sufficiency holds for all formulas and therefore $\suff$ does not depend on $\phi$.
  In our learning scenario this means that we can restrict the search for consistent parameters to such a sufficient set.
  For every $\train$, the set $V(\tree)$ is clearly $1$-sufficient as all parameters are nodes of $\tree$.
  We can always choose $\phi'$ such that it does not contain subformulas of the form $R(y_i,y_j)$ or $R'(y_i)$, where $y_i,y_j$ are parameters.
  In that case only the relations between parameters and examples have an impact on the consistency of the corresponding hypothesis.

  We say that for a signature $\sigma$, two nodes $v_1,v_2\in V(\tree)$ share the same \emph{relative position} with respect to a training set $\train$ if for every binary $R\in\sigma$ and every $x\in \train$ we have that $R(x,v_1) \equiv R(x,v_2)$ and $R(v_1,x) \equiv R(v_2,x)$.
  Two nodes $u,v$ with the same relative position to $\train$ cannot be distinguished by a quantifier-free formula and thus any set containing at least one representative of every relative position is $1$-sufficient. 
  Note that for every $\phi \in \Phi[k,\ell]$ with parameter $\bar v$ we can find a formula $\phi'\in\Phi[k,\ell]$ such that $\phi'(u,\bar v')$ accepts the exact same examples from $\train$ as $\phi(u,\bar v)$ given that $\bar v$ and $\bar v'$ share the same relative position to $\train$.
  Again this holds as the classification of $u$ only depends on the relations between $u$ and the parameters which by definition are identical for $\bar v$ and $\bar v'$.

  Let $\suff$ be defined as the $2$-neighborhood of $\text{LCA}(\train)$, where $\text{LCA}(\train)$ is the closure of $\train$ under lowest common ancestors.
  For unranked trees we use the local $2$-neighborhood of a node $u\in \tree$ for $\suff$, defined in Section \ref{subsec:access_model} as the set containing $u$ and the parent, first child, as well as left and right sibling of $u$.
  The set $\suff$ is linear in $|\train|$ and its sufficiency can be shown by a case-distinction outsourced in the following two claims whose proofs are deferred until after this proof.
  \begin{claim}\label{claim:sufficient_binary}
    $\suff$ is $1$-sufficient for $\train$ on binary trees.
  \end{claim}  

  \begin{claim}\label{claim:sufficient_unranked}
    $\suff$ is $2$-sufficient for $\train$ on unranked trees.
  \end{claim}

  In order to find a consistent hypothesis, the learning algorithm tests every possible quantifier-free formula for every parameter setting $\bar v \in (\suff)^\ell$ for consistency with $\train$.
  It is bound to find a consistent hypothesis since $\suff$ is sufficient.
  The set $\suff$ can be computed in time $\onot(|\train|^2)$ using extended local access.
  Since each evaluation of a quantifier-free formula needs constant time, this leads to the runtime $\onot\left(|\train|^2+|\train|^{\ell}|\train|\right)$ stated in the theorem.
\end{inlineproof}

\begin{inlineproof}[Proof of \cref{claim:sufficient_binary}]
  For every node $x \in V(\tree)$, we consider the following sets of nodes describing its relative position to $\train$. 
  \begin{align*}
    \examplesFamily &=\{x \mid x \in \train, E_1(v,x) \lor E_2(v,x) \lor E_1(x,v) \lor E_2(x,v)\};\\
    \examplesAbove  &=\{x \mid x \in \train, x\leq v)\};\\
    \examplesBelow  &=\{x \mid x \in \train, v\leq x)\}. 
  \end{align*}

  Sufficiency of $\suff$ means that if we have a consistent hypothesis, there is also one that uses only parameters from $\suff$, possibly based on a different formula.
  Let $\phi \in \Phi_\sigma[k,\ell]$ be a quantifier-free formula and $\bar v\in V(\tree)^\ell$ a parameter vector such that $\model{\tree}{\phi}{x}{\bar v}$ is consistent with $\train$.
  Let $v=v_i \in \bar v$ be a parameter from $\bar{v}\in V(\tree)^\ell$.
  We prove that there exists a node $w$ in $\suff$ such that $\model{\tree}{\phi'}{x}{v_1,\dots,v_{i-1},w,v_{i+1},\dots,v_\ell}$ is consistent if $\model{\tree}{\phi}{x}{v_1,\dots,v_\ell}$ is.
  That means that we can essentially substitute $v_i$ by $w$, inducing small changes on $\phi$, without changing the acceptance behavior for any $x\in \train$.

  If $S_\text{family}(v)\neq\emptyset$, then $v\in \suff$ and we define $w=v$. 
  In the following case distinction we therefore have a distance of at least $2$ to every example.
  The last two cases are illustrated in \cref{fig:quantifierFree:aboveAndBelow}.

  \begin{itemize}
    \item If $\examplesAbove = \examplesBelow{} = \examplesFamily{} = \emptyset$, the node $v$ is unrelated to every example node. 
      In that case the parameter $v_i$ has the same influence on every $x\in \train$ and we can replace all atoms containing $y_i$ by the boolean $\false$.
      Alternatively we could add an additional node $\bot$ which occurs in no relation and assign $w=\bot$.
    \item If $\examplesBelow{}\neq\emptyset$, we choose $w$ as the lowest common ancestor of $\examplesBelow{}$ which clearly is in $\suff$ since $\examplesBelow \subseteq \train$.
    \item If $\examplesBelow{} = \emptyset$ but $\examplesAbove\neq\emptyset$, we let $w$ be the lowest node of $\examplesAbove$. 
      In this case we have to replace every atom of the form $(\neg x \leq y_i)\equiv y_i>x$ by $\false$. 
      This step assumes that $\phi$ is given in negation normal form (NNF), that is negations only appear at the atoms. 
      The substitution is necessary as we moved $v$ upwards and thus possibly $S_\text{below}(w)\neq \emptyset$.
  \end{itemize}

  \begin{figure}
    \begin{minipage}{0.5\textwidth}
      \begin{tikzpicture}
  \node[] (root) at (8,8) {root};
  \node[] (middle1) at (6.6,6.6) {$x$};
  \node[] (middle2) at (5.8,5.8) {$x$};
  \node[] (w) at (5,5) {$x$};
  \node[right=of w] (newText) {$w$};
  \draw[<-] (w) -- (newText);
  \node[] (v) at (4,4) {$v$};
  \node[] (x) at (6.75,2.75) {$x$};
  \draw[decorate,decoration={snake,post length=0mm}] (root) -- (middle1) -- (middle2) -- (w);
  \draw[decorate,decoration={snake,post length=0mm}] (w) -- (v);
  \draw[decorate,decoration={snake,post length=0mm}] (w) -- (x);
  \node[] (text) at (4,2) {no examples};
  \draw[] (v) -- (text.north west);
  \draw[] (v) -- (text.north east);
\end{tikzpicture}
     \end{minipage}    \begin{minipage}{0.5\textwidth}
      \usetikzlibrary{decorations,decorations.pathmorphing}

\begin{tikzpicture}
  \node[] (root) at (8,8) {root};
  \node[] (middle1) at (6.6,6.6) {$x$};
  \node[] (middle2) at (5.8,5.8) {$x$};
  \node[] (v) at (5,5) {$v$};
  \node[] (w) at (4,4) {$w$};
  \draw[decorate,decoration={snake,post length=0mm}] (root) -- (middle1) -- (middle2) -- (v);
  \draw[decorate,decoration={snake,post length=0mm}] (v) -- (w);
  \node[] (text) at (4,2) {some examples};
  \draw[] (w) -- (text.north west);
  \draw[] (w) -- (text.north east);

  \node[right=of w] (newText) {$\text{LCA}(S_\text{below}(v))$};
  \draw[<-] (w) -- (newText);
\end{tikzpicture}
     \end{minipage}
  
  \caption{Left: The case that all related examples are above the parameter $v$. Right: The setting that there are examples in the subtree below $v$.}\label{fig:quantifierFree:aboveAndBelow}
  \end{figure}
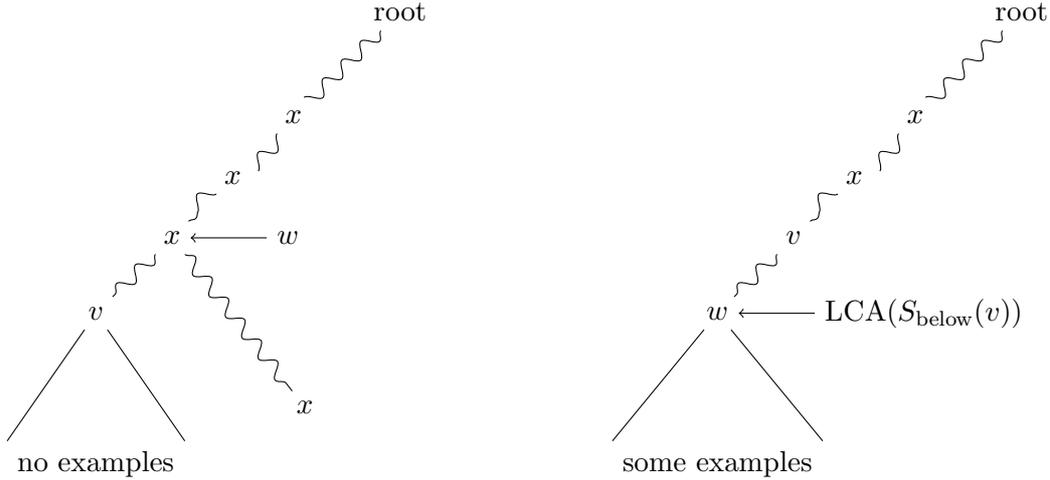
  This shows the claim.
\end{inlineproof}

\begin{inlineproof}[Proof sketch for \cref{claim:sufficient_unranked}]
  For every node $x \in V(\tree)$, we consider the following sets of nodes describing its relative position to $\train$ in an unranked trees. 
  \begin{align*}
    \examplesFamily &=\{x \mid x \in \train, E_1(v,x) \lor E_2(v,x) \lor E_1(x,v) \lor E_2(x,v)\};\\
    \examplesAbove  &=\{x \mid x \in \train, x\leq v)\};\\
    \examplesBelow  &=\{x \mid x \in \train, v\leq x)\}; \\
    \examplesLeft   &=\{x \mid x \in \train, v\prec x\}; \\
    \examplesRight  &=\{x \mid x \in \train, x\prec v\}.
  \end{align*}
  
  Let $v=v_i \in \bar v$ be a parameter from $\bar{v}\in V(\tree)^\ell$.
  We prove that there exist up to two nodes $w,w_2\in\suff$ simulating the same relative position in a slightly adapted quantifier-free formula.
  More formally, there exist $w,w_2\in \suff$ and $\phi'$ such that $\model{\tree}{\phi'}{x}{v_1,\dots,v_{i-1},w,v_{i+1},\dots,v_\ell}$ is consistent if $\model{\tree}{\phi}{x}{v_1,\dots,v_\ell}$ is.

  If $S_\text{family}(v)\neq\emptyset$, then $v\in \suff$ and we define $w=v$.
  Otherwise we consider the relations $\leq$ and $\prec$ independently by choosing $w$ to be in the same relative position as $v$ with regard to $\leq$ and $w_2$ to be in the same relative position as $v$ with regard to $\prec$.
  We modify $\phi$ in order to make use of  $w,w_2$ and therefore introduce an additional free variable $y'_i$.
  The new formula $\phi'$ is generated from $\phi$ by the the following substitutions of atoms in $\phi$, again we assume $\phi$ to be in NNF:
  \begin{itemize}
    \item Substitute $\leq(x,y_i)$ or $\leq(y_i,x)$ by $R(x,y'_i)$ and
    \item Substitute $R(x,y_i)$ or $R(y_i,x)$ with $R\in\{=,E_2,E_1\}$ by $\false$.
  \end{itemize}

  For $\leq$ we choose $w$ as in \cref{claim:sufficient_binary}.
  For $\prec$ we have the following case distinction.
  \begin{itemize}
    \item If $\examplesLeft = \examplesRight = \emptyset$, then none of the examples is a sibling and we can substitute every atom of the form $x \prec y_i'$ and $y_i' \prec x$ by $\false$.
    \item If $\examplesLeft \neq \emptyset$, then if the siblings of $v$ are numbered $u_1,\dots, u_k$ and $j$ is the largest index such that $u_j\in\examplesLeft$, we choose $w_2=u_{j+2}$. 
      This node $u_{j+2}$ exists since $v$ is not a direct neighbor of any $u\in \train$ and $\examplesLeft\neq \emptyset$.
    \item Otherwise $\examplesRight \neq \emptyset$. 
      Then let $u_1,\dots,u_k$ be the siblings of $v$ and $j$ the smallest index such that $u_j\in \examplesRight$ and choose $w_2 = u_{j-2}$.
      The node $u_{j-2}$ exists by the same argument as in the previous case.
  \end{itemize}

  Note that if either $\examplesFamily \neq \emptyset$, $\examplesLeft = \examplesRight = \emptyset$ or $\examplesAbove = \examplesBelow = \emptyset$ we only need a single parameter.
\end{inlineproof}

\newcommand{\monoid}{M}
\newcommand{\monoidmult}{\cdot_M}
\newcommand{\hp}{\text{hp}}    \newcommand{\decomp}{\hp}
\newcommand{\dfa}{\mathcal{A}}
\newcommand{\dta}{\mathfrak{A}}
\newcommand{\rightChild}{\text{R}}
\newcommand{\leftChild}{\text{L}}
\newcommand{\lr}{\braces{\leftChild{},\rightChild{}}}
\newcommand{\leaf}{\text{leaf}}
\newcommand{\parameters}{2^{\{y_1,\dots,y_\ell\}}} \newcommand{\funcAlphabetExtensionForDfa}{f_\rho}
\newcommand{\algo}{L}

\newcommand{\innermonoid}{\mathcal{M}}
\newcommand{\outermonoid}{{\hat{\mathcal{M}}}}
\newcommand{\freemonoid}{M_\text{free}}
\newcommand{\alphabetInnerMonoidStrings}{\Sigma \times \{P,N,?\} \times 2^{\{y_1,\dots,y_\ell\}}}
\newcommand{\freeToInnerMorphism}{h}
\newcommand{\freeToOuterMorphism}{\hat h}
\newcommand{\parameterSelection}{\project_{\bar y}}
\newcommand{\monoidParameters}{p}
\newcommand{\stateParameters}{p'}
\newcommand{\factorizationTree}{{F}} \newcommand{\acceptingStates}{F}
\newcommand{\project}{\pi}

\section{Monadic Second-Order Logic} \label{sec:MSO}

In this section we consider learning of unary MSO formulas on trees.
The ordering relation~$\leq$ used in the previous chapter can be expressed in MSO, thus we do not include it in the signature.
Given a binary tree $\tree$ we want to learn a unary relation $R\subseteq V(\tree)$ using an MSO formula $\phi(x\smid y_1,\dots,y_\ell)$ and $\ell$ parameters $v_1,\dots,v_\ell \in V(\tree)$. 
We aim for simple concepts and therefore restrict the number of parameters $\ell$ and the quantifier rank $q$ of $\phi$.
Let $\mso[q,\ell+1]$ be the class of MSO formulas with $\ell+1$ free variables and quantifier rank up to $q$.
Note that the class $\mso[q,\ell+1]$ is finite up to equivalence \cut{(renaming variables, reordering terms, copying terms) }for every fixed $q$ and $\ell$, such that we can iterate over all formulas from this class to find a consistent one.

The algorithm we present separates the task of analyzing the background structure and building an {index} structure from the task of finding a consistent hypothesis for a given set of training examples based on that index.
On the one hand, this emphasizes the challenges that arise when learning formulas.
On the other hand, the indexing phase can be completed before we know a single example and we can thus reuse it for similar learning tasks on the same dataset.
Formally, we have an indexing phase in which the learning algorithm $\algo$ has local access to the tree~$\tree$, but not to the training set $\train$, and produces an index $I(\tree)$.
In the learning phase, $\algo$ gets $\train$ and has local access to $\tree$ and $I(\tree)$.
Based on that input, it produces a formula $\phi(x\smid \bar y)$ and a parameter configuration $\bar v$ such that $\model{\tree}{\phi}{x}{\bar v}$ is consistent with $\train$.
We call such an algorithm an \emph{indexing algorithm}, the time needed to build the index \emph{indexing time} and the time for the actual search of a consistent hypothesis based on that index \emph{search time}.
We show the following theorem which implies \cref{thm:main} stated in the introduction.

\begin{theorem}\label{thm:msoLearning}
  Let $q,\ell\in \NN$ be fixed.
  Given a tree $\tree$ and a training set $\train$ that is consistent with an MSO formula $\vartheta(x\smid\bar y)\in\mso[q,\ell+1]$, it is possible to find a formula $\phi(x\smid \bar y)\in\mso[q,\ell+1]$ that is consistent with $\train$ in indexing time $\onot(|\tree|)$ (without access to~$\train$) and search time $\onot(|\train| \log |\tree|)$.
\end{theorem}

Technically, we only consider binary trees.
The theorem directly extends to unranked trees as there exist MSO interpretations between ranked and unranked trees increasing the quantifier-rank of the a consistent formula by at most $1$.

Instead of solving the model learning problem from \cref{thm:msoLearning} directly, we give an algorithm for the corresponding parameter learning problem in \cref{thm:msoParameterLearning}.
Together with an exhaustive search over all (finitely many) semantically different MSO formulas with at most $\ell+1$ free variables and quantifier rank up to $q$, this proves the theorem.
In the parameter learning problem the learning algorithm has access to $\varphi$ and has to find a parameter setting $\bar v$ such that $\model{\tree}{\varphi}{x}{\bar v}$ is consistent with $\train$.
For a fixed formula $\phi$, the presented algorithm creates an index structure of size $\onot(|\tree|)$ in time $\onot(|\tree|)$.
In time $\onot(\log(|\tree|))$, it then computes consistent parameters $\bar v$ for $\phi$ or outputs that there is no such $\bar v$.

For the parameter learning problem there is a linear lower bound \cite{grorit17} for learning on strings which trivially extends to trees.
The running time of the presented algorithm therefore differs from the optimal one only by a logarithmic factor. 
The number of formulas tested in the exhaustive search is non-elementary in $q$ and $\ell$ such that the presented algorithm is of primarily theoretical interest.

\begin{theorem}\label{thm:msoParameterLearning}
  Let $q,\ell\in \NN$ be fixed.
  Given a tree $\tree$, an MSO formula $\phi(x\smid \bar y) \in\mso[q,\ell+1]$ and a training set $\train$ consistent with $\phi$.
  Then a set of parameters $\bar v \in V(\tree)^\ell$ such that $\model{\tree}{\phi}{x}{\bar v}$ is consistent with $\train$ can be found in indexing time $\onot(|\tree|)$ and search time $\onot(\log |\tree|\cdot |\train|)$ \end{theorem}

Although $\phi$ is part of the input in \cref{thm:msoParameterLearning} its quantifier rank is explicitly bounded in order to achieve the desired runtime bound.
We use automata based techniques in our proof and thus need to restrict the size of a DFA that is equivalent to $\phi$.
In general the size of that DFA is bounded by a tower of twos linear in the number of quantifier alternations in $\phi$.
By bounding the quantifier rank of $\phi$ we replace this non-elementary dependence on the length of the input by a constant.

\subsection{Decomposing trees into strings}\label{subsec:treeDecomp}

In the proof of \cref{thm:msoParameterLearning} we apply techniques based on automata theory and monoids where the latter require the input to consist of strings instead of trees.
We therefore decompose our background tree into a set of strings by \emph{heavy path decompositions} introduced by Harel et al. \cite{harel1984fast}.
We start with some additional notation.
Given a tree $\tree = (V(\tree),E(\tree),\order)$
and a partition $P_1,\dots,P_p$ of $V(\tree)$. Let $\tree[P_i]$ be the induced substructure of $P_i$ on $\tree$ defined by restricting $V(\tree), E(\tree)$ and $\order$ to the elements of $P_i$. 
We consider only induced substructures where $\tree[P_i]$ is a string which implies that $\order$ restricted to $P_i$ is a total order.
Let furthermore $\tree[u]_\order$ be defined as $\tree[\braces{w \mid u\order w}]$, that is, we restrict $\tree$ to the induced subtree rooted at $u$.

\begin{figure}[t]
  \begin{center}
    \vspace{-2ex}

\colorlet{pathblue}{rgb:cyan,2;black,4;blue,3}
\colorlet{pathorange}{rgb:orange,3;black,0;red,2;yellow,1}
\begin{forest}
   for tree={l sep=1em, l=1em,s sep=1.1em, math content, inner sep=0.25em,draw,circle,thick,edge=thick},
   [,tikz={\node[left=0.1cm] at () {$u_{r}$};}
    [,edge=dotted,tikz={
      \node [draw, inner sep=0.5em, fit=()(!11)(!1ll),pathblue,rounded corners,thick]{};
      \node [right=0.1cm,yshift=0.05cm,pathblue] at () {$u'$};
      \node [below left=-0.1cm and 0.25cm,pathblue] at () {$P_2$};
      },for tree={pathblue},for descendants={edge={pathblue}}
      [
        [,edge=dotted]
        [[][,edge=dotted]]
      ]
      [,edge=dotted]
    ]
    [,tikz={\node [below left=0.3cm and -0.1cm,pathorange] at () {$w$};},pathorange,name=a
      [,pathorange,edge=pathorange,name=b
        [,edge=dotted,name=c]
        [,pathorange,edge=pathorange,
          [[][,edge=dotted]]
          [,edge=dotted[,name=e][,edge=dotted,name=f]]
        ]
      ]
      [,edge=dotted,name=g[][,edge=dotted,name=d]]
    ]
   ]
   \draw[pathorange,thick] \convexpath{a,g,d,f,e,c,b}{7pt};
\end{forest}

     \vspace{-2ex}
  \end{center}
  \caption{Heavy path decomposition of a tree. Solid lines indicate the different heavy paths, the cut-off subtree of $u_r$ is shown in blue and the dependent subtree of $w$ is shown by the orange shape.}\label{fig:heavyPathDecomposition}
\end{figure}
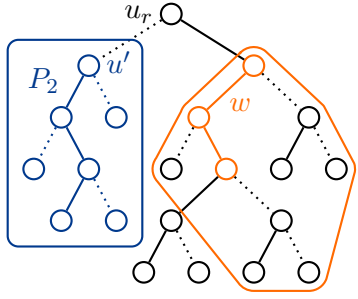

In the following we present the definition of the heavy path decompositions as introduced in \cite{harel1984fast}.
We say that a node $u\in V(\tree)$ is at least as \emph{heavy} as $v\in V(\tree)$ \cut{($u \geq_h v$)} if $|\{ w \mid u\order w\}| \geq |\{ w \mid v\order w\}|$.
That means that $u$ is heavier than $v$ if there are more nodes below $a$ than there are below $b$.
The \emph{heavy path} of $\tree$ at $u\in V(\tree)$ is a path $(u_0,u_1,\dots,u_\ell)$ such that $u_0=u$, the node $u_\ell$ is a leaf and $u_i$ is the leftmost child of $u_{i-1}$ that is at least as heavy as all other children of $u_{i-1}$ for every $i\leq \ell$.
Let $\tree$ be a binary tree and $P= (u_0,\dots,u_\ell)$ a heavy path.
The \emph{cut-off subtree} at a node $u_i\in P$ with a child $u'\not\in P$ is $\tree[u']_\order$. 
The \emph{dependent subtree} of a substring $w\subseteq P$ consists of $w$ and the union of the cut-off subtrees for every $u\in w$.
This differs from $\tree[u_0]_\order$ where $u_0$ is the lowest node of $w$ in $\tree$ if $w$ does not contain a leaf of $\tree$. 
The heavy path decomposition $\decomp{}(\tree) = \braces{P_1,\dots,P_p}$ of $\tree$ is constructed by first computing $P_1$ as the heavy path of $\tree$ at its root and then recursively computing the heavy paths of the cut-off subtrees for each $u\in P_1$.

\cref{fig:heavyPathDecomposition} illustrates the heavy path decomposition as well as cut-off and dependent subtrees.
From a simple counting argument we get the following lemma on heavy path decompositions.

\begin{lemma}[\cite{BalminPV04}]\label{lem:logarithmicHeavyPathsToTheRoot}
Every path from the root of $\tree$ to a node $u\in V(\tree)$ intersects with at most $\log |\tree|$ different heavy paths from $\decomp(\tree)$.
\end{lemma}

\subsection{Simulating tree automata by string automata}\label{subsec:DfaConstruction}

We now describe how we can use path decompositions such as $\decomp{}(\tree)$ and deterministic finite automata (DFAs) to simulate a deterministic bottom-up tree automaton (DTA) $\dta$ on $\tree$ as shown in \cite{BalminPV04}.
Let $\tree$ be a tree, $\decomp{}(\tree)\! = \!\{P_1,\dots,P_p\}$ its heavy path decomposition and $\dta$ a DTA with statespace $Q$.
Let $\rho$ be the run of $\dta$ on $\tree$ and $Q' = Q \,\dot\cup\, \braces{q_0}$.
We define a DFA $\dfa$ and a labeling function $\funcAlphabetExtensionForDfa\colon V(\tree) \rightarrow \Sigma\times Q' \times \lr$ extending the labels of $V(\tree)$ by annotations about state and direction of the cut-off subtree.
These labels, given by $\funcAlphabetExtensionForDfa{}$, allow us simulate the tree automaton $\dta$ on $T$ by running a DFA $\dfa$ on $\decomp(\tree)$.

\begin{lemma}\label{lem:sameStatesDtaDfa}
  Let $\rho$ be the run of the DTA $\dta$ on $\tree$ and $\tree'$ the tree we get from applying $\funcAlphabetExtensionForDfa{}$ to every node of $\tree$.
  Then there is a DFA $\dfa$ such that for every $u\in P_i$, $\dfa$ is in state $q$ after reading $\tree'[P_i]$ up to position $u$, if and only if $\rho(u)=q$.
\end{lemma}

\begin{inlineproof}
  To match the evaluation order of bottom-up tree automata, the constructed DFA $\dfa$ reads paths $\tree'[P_i]$ in reversed order, that is, starting from the leaves and proceeding towards the root of~$\tree$.
  The labeling function $\funcAlphabetExtensionForDfa{}$ uses the run $\rho$ of $\dta$ on $\tree$:
  For a leaf $u$ with label $a$ in $\tree$ we have $\funcAlphabetExtensionForDfa(u) = (a,q_0,\leftChild{})$.
  For an inner node $u\in P_i$ with label $a$ and children $u'\in P_i$ and $u''\not\in P_i$ where $u''$ is a right child of $u$ we have $\funcAlphabetExtensionForDfa{}(u) = (a,\rho(u''),\rightChild{})$ and correspondingly $(a,\rho(u''),\leftChild{})$ if $u''$ was a left child of $u$.
  This means that the function $\funcAlphabetExtensionForDfa{}$ stores the state of the DTA $\dta$ on the cut-off subtree at $u$ in the label of $u$.
  \cref{fig:dfaDtaSimulation} shows an example for $\funcAlphabetExtensionForDfa{}$ given a tree $\tree$ and a run $\rho$ of a DTA on $\tree$.

  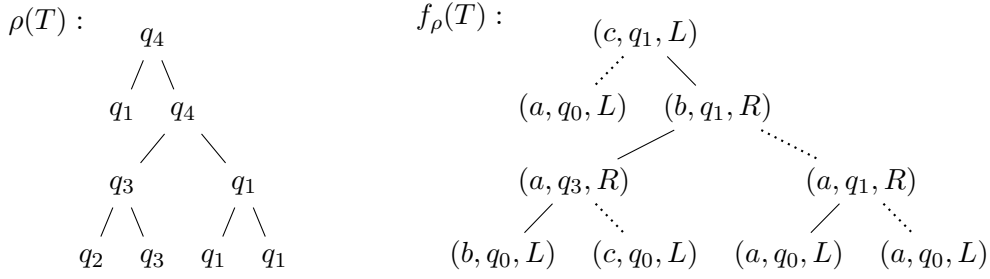
\begin{figure}[t]
    \begin{center}
      \vspace{-2ex}

\begin{forest}
  for tree={l sep=1.0em, l=1em, math content, inner sep=0.25em,minimum height=1.5em},
  [,phantom,s sep=5em
  [q_4,tikz={\node [above left=-0.3em and 2em]at () {$\rho(T):$};},tier=oben
    [q_1]
    [q_4
      [q_3[q_2][q_3]]
      [q_1[q_1][q_1,tier=unten]]
    ]
  ]
  [{(c,q_1,L)},tikz={\node [above left=-0.3em and 4.5em]at () {$f_\rho(T):$};},tier=oben
    [{(a,q_0,L)},edge={dotted,thick}]
    [{(b,q_1,R)}
      [{(a,q_3,R)}
        [{(b,q_0,L)}]
        [{(c,q_0,L)},edge={dotted,thick}]
      ]
      [{(a,q_1,R)},edge={dotted,thick}
        [{(a,q_0,L)}]
        [{(a,q_0,L)},edge={dotted,thick},tier=unten]
      ]
    ]
  ]
  ]
\end{forest}
       \vspace{-2ex}
    \end{center}
    \caption{Applying the label transformation $\funcAlphabetExtensionForDfa{}$ to $\decomp{}(\tree)$ for some tree $\tree$}
    \label{fig:dfaDtaSimulation}
  \end{figure}

  Let $\dta = (Q,\Sigma, \delta, F)$ be a DTA and $\decomp(\tree)= \!\{P_1,\dots,P_p\}$ the heavy path decomposition of $\tree$.
  For the case distinction we assume that the transition relation $\delta$ of $\dta$ is split into $\delta_0\colon \Sigma\rightarrow Q$ for leaf nodes and $\delta_2\colon \Sigma\times Q\times Q \rightarrow Q$ for inner nodes.
  We now define the DFA $\dfa = (Q',\Sigma\times Q' \times \lr,\delta', q_0, F)$ where $Q' = Q \,\dot\cup\, \braces{q_0}$ with $Q,\Sigma$ and $F$ from $\dta$.
  The transition relation $\delta'$ of the DFA $\dfa$ is defined as follows:
  \begin{alignat*}{2}
    \delta'(q_0,(a,q_0,d)) &= \delta_0(a) &&~\text{ for all } a\in\Sigma, d\in \{\leftChild,\rightChild\},\\
    \delta'(p,(a,q,\leftChild)) &= \delta_2(a,q,p) \qquad&&~\text{ for all } a\in \Sigma,~ p,q\in Q,\\
    \qquad\delta'(p,(a,q,\rightChild)) &= \delta_2(a,p,q) &&~\text{ for all } a\in \Sigma,~ p,q\in Q.
  \end{alignat*}

  Lemma \ref{lem:sameStatesDtaDfa} holds as at every position $u\in P_i$ the DFA $\dfa$ has access to the same information as $\dta$ through the extended alphabet and thus can assign the same state.
\end{inlineproof}

Let $<_\hp$ be a partial order on the heavy paths of $\decomp{}(\tree)$ with $P_j<_\hp P_i$ if $P_j$ is part of the dependent subtree of $P_i$.
This way $<_\hp$ encodes dependencies between the heavy paths.
For every node $u\in P_i$ with cut-off child $u'\in P_j$, the state $\rho(u')$ only depends on $P_j$ with $P_j <_\hp P_i$.
Hence, by using \cref{lem:sameStatesDtaDfa} we can recursively compute $\rho$ using $\dfa$ on $\decomp{}(\tree)$.
We say that $\dfa$ accepts a tree $\tree$ if it accepts the string $\tree[P_1]$ where $P_1\in\decomp{}(\tree)$ contains the root of $\tree$.

\begin{corollary}\label{cor:dfaEquivalentToDta}
  When evaluating $\dfa$ on $\decomp{}(\tree)$ in an order consistent with $<_\hp$ for every inner node $u\in V(\tree)$ the state $\rho(u')$ of $\dta$ on the cut-off subtree $\tree[u']_\leq$ at $u$ is known when running $\dfa$ on the heavy path $P_i$ with $u\in P_i$.
  Hence we can compute $\rho$ on-the-fly using $\dfa$ on $\decomp{}(\tree)$.
\end{corollary}

This means that whenever we evaluate $\dfa$ on some $P_i$, all the labels $\funcAlphabetExtensionForDfa(u)$ in $P_i$ can be computed by running $\dfa$ on the paths $P_j$ with $P_j <_\hp P_i$ and incorporating this information in $\funcAlphabetExtensionForDfa{}$.

\begin{inlineproof}
  For every node $u\in V(\tree)$, the labeling function $\funcAlphabetExtensionForDfa{}$ depends only on the previous label, the state $\rho(u')$ where $u'$ is the root of the cut-off subtree at $u$ and whether $u'$ is a left or right child of $u$.
  For every path consisting just of a leaf $u$ $\dfa$ assigns the same state as $\dta$ since the state only depends on the label of $u$ in $\tree$.
  For $u\in P_i$ and $u'\in P_j$ where $u'$ is a child of $u$ and $i\neq j$, that is, $u'$ is the root of the cut-off subtree, we have that $P_j <_\hp P_i$.
  Therefore, when evaluating the heavy paths in an order consistent with $<_\hp$ by \cref{lem:sameStatesDtaDfa} $\rho(u')$ is known by the time $P_i$ is evaluated such that we can compute $\rho$ on-the-fly using $\dfa$ on $\decomp(\tree)$.
\end{inlineproof}

\subsection{Monoids and factorization trees}

A monoid $\mathcal M = \{ \monoid, \monoidmult, 1_\monoid\}$ consists of a set of elements $\monoid$ and an associative multiplication operator $\monoidmult$ with neutral element $1_\monoid$.
We usually refer to the monoid $\mathcal M$ by its set of elements $\monoid$ and write $m_1m_2$ or $m_1\cdot m_2$ for $m_1 \monoidmult m_2$.
A monoid morphism $h\colon \monoid \rightarrow \monoid'$ is a function that translates a monoid in a consistent way, meaning that we have $h(m_1\cdot m_2) = h(m_1)\cdot h(m_2)$ and $h(1_\monoid) = 1_{\monoid'}$.

The free monoid $\freemonoid^\Sigma = \{ \Sigma^*, \cdot_\text{free}, \epsilon\}$ consists of all finite strings over $\Sigma$ with concatenation as multiplication and the empty word $\epsilon$ as neutral element.
We write $\Sigma^*$ and $w\in \Sigma^*$ for $\freemonoid{}^\Sigma$ and $m_w \in \freemonoid{}^\Sigma$ in monoid morphisms, that is we again identify monoids with their elements. 
A language $\mathcal L$ over $\Sigma$ is \emph{finitely monoid recognizable} if there is a finite monoid $\monoid$, a monoid morphism $h\colon \freemonoid \rightarrow \monoid$ and a subset $F\subseteq \monoid$ such that $w\in \mathcal L$ if and only if $h(w)\in \acceptingStates$.
A language $\mathcal L$ is finitely monoid recognizable, if and only if it is regular \cite{RabinS59}.
The \emph{transition monoid} $\monoid$ of a DFA $\dfa$ with state space $Q$ is defined as functions $m\colon Q \rightarrow Q$ with composition as multiplication, and the identity as neutral element. 
The corresponding monoid morphism $h_M$ maps a string $w\in \Sigma^*$, that we identify with the corresponding element from $\freemonoid{}^\Sigma$, to the function $m\colon Q\rightarrow Q$ modeling the effect of reading $w$ on the states of $\dfa$.
A monoid element $m$ is \emph{productive} if there are $m',m''$ with $m'\cdot m\cdot m''\in \acceptingStates$. For further details on the equivalence between finite monoids and DFAs see for example \cite{Colcombet11,Bojanczyk12}.

Every path $P_i\in \decomp{}(\tree)$ induces a string $\tree[P_i]$ over $\Sigma$ which we interpret as an element $m\in \freemonoid{}^\Sigma$ of the free monoid over the tree alphabet $\Sigma$. 
In the following we will not explicitly distinguish between the path $P_i$, the corresponding string $\tree[P_i]$ and the monoid element $m$.

A \emph{factorization tree} of a sequence of monoid elements $s=m_1,m_2,\ldots,m_n$ from a monoid $\monoid$ is a tree $F_s$ where each node $u\in V(F_s)$ is labeled by an element $m_u$ from $M$ and the sequence of labels at its leaves is $s$.
For each inner node $u\in V(F_s)$ labeled by $m_u$ with children $u_1,\dots, u_i$ labeled by $m_{u_1},\dots,m_{u_i}$ we have $m_u = m_{u_1}m_{u_2}\!\cdots m_{u_i}$.
We can always choose a balanced binary tree for $F_s$ such that we get a factorization tree of height $\lceil\log |s|\rceil$.
An alternative are \emph{Simon factorization trees}, where each node $u$ is either a leaf, a binary node or satisfies the property that all child nodes $u_1,\dots,u_n$ of $u$ share the same label $m_u$ with $m_u = m_u\cdot m_u$.
We call such an $m_u$ idempotent.
\begin{theorem}[Simon \cite{Simon90}]\label{thm:simon}
  For every sequence $s = [m_1, m_2,\dots, m_n]$ of monoid elements from $M$, there is a Simon factorization tree of height at most $3|M|$. 
  This factorization tree can be computed in time $n\cdot \text{poly}(|M|)$.
\end{theorem}

In a set of \emph{updates} $\labelChanges{} = \{ (i_1,m_1),\dots,(i_p,m_p)\}$  with $(i_j,m_j) \in \NN\times \monoid$ each tuple consists of an index $i_j$ and the new monoid  element $m_j$ for this index.
Applying $\labelChanges{}$ to a sequence $s$ of monoid elements results in the sequence $s^U$ with $s^\labelChanges{}[i_j]=m_j$ for $j\in \{1,\dots,s\}$ and $s^U[j]=s[j]$ for all $j\not\in\{i_1,\dots,i_p\}$ where $s[i]$ denotes the $i$-th element of the sequence $s$.
We use the following lemma to apply such updates in Simon factorization trees. 

\begin{lemma}[\cite{DBLP:conf/alt/GroheLR17}]\label{lem:updateFactorizationTree}
  Given a Simon factorization tree $\factorizationTree{}$ of height $h$ over a sequence $s$ of elements from $\monoid$ and a set $\labelChanges{}$ of updates.
  There is an algorithm returning a Simon factorization tree of height $2h+3|M|$ for the updated sequence $s^\labelChanges{}$ in time $\onot(|\labelChanges{}|h+|\labelChanges{}||M|)$.
\end{lemma}

\subsection{Constructing the monoids and monoid morphisms for MSO learning}

In the following we construct and analyze the monoid structure used in the learning algorithm for \cref{thm:msoParameterLearning}.
Let $\tree$ be a tree over $\Sigma$ and $\tree_1$ the same tree over $\Sigma_1 = \Sigma\times 2^{\{y_1,\dots,y_\ell\}}\times \{?,N,P\}$ including information about the free variables of a formula $\psi(P,N,\bar y)$ in the labels.
Let $\Sigma_2 = \Sigma_1 \times Q' \times \lr$ be the alphabet from \cref{subsec:DfaConstruction} which extends $\Sigma_1$ in order to work with string automata instead of tree automata. 
Let $\funcAlphabetExtensionForDfa{}: V(\tree) \rightarrow \Sigma_2$ be the labeling function used in \cref{lem:sameStatesDtaDfa} and let $\tree_2$ be the tree that results from $\tree$ by labeling its nodes using~$\funcAlphabetExtensionForDfa{}$.
Note that the concrete states $Q'$, and therefore the whole construction, depend on $\psi(P,N,\bar y)$, the formula that checks whether there are parameters $\bar v$ for $\bar y$ such that $\phi$ is consistent with $\train$.

\begin{lemma}\label{lem:monoidsInnerAndOuter}
  Let $\phi(x\smid \bar y)$ be an MSO formula, $\tree$ a tree over $\Sigma$ with $\decomp{}(\tree) = \braces{P_1,\dots,P_n}$ and $\train$ a training set.
  Then there exists
  \begin{itemize}
    \item an alphabet $\Sigma_3$,
    \item a construction transforming $\tree$ to a tree $\tree_3$ over $\Sigma_3$     \item a monoid $\outermonoid{}$ with a set $F\subseteq \outermonoid{}$ of final elements 
    \item and a monoid morphism $\freeToOuterMorphism{}: \Sigma_3^* \rightarrow \outermonoid{}$ 
  \end{itemize}
  such that $\freeToOuterMorphism{}(\tree_3[P_1]) \in F$ if and only if there exists $\bar v \in V(\tree)^\ell$ such that $\model{\tree}{\phi}{x}{\bar v}$ is consistent with~$\train$.
\end{lemma}

\begin{inlineproof}
  We first describe the construction of the monoid $\outermonoid{}$, the alphabet $\Sigma_3$ and the monoid morphism $\freeToOuterMorphism{}$ in five steps and then prove the lemma based on that description.
  The whole construction is based on $\psi$ depending on $\phi$.
  \[
    \psi(P,N,\bar y) = \forall x (P(x)\rightarrow \phi(x,\bar y)) \land (N(x) \rightarrow \neg \phi(x,\bar y))
  \]
  The formula $\psi(P,N,\bar y)$ checks for a given training set $\train$, split into positive ($P$) and negative examples ($N$), and a parameter vector $\bar v\in V(\tree)^\ell$ whether the hypothesis $\model{\tree}{\phi}{x}{\bar v}$ is consistent with~$\train$.
  This is essentially what we are going to test, but for a better runtime bound, we transform the formula into a monoid and an algorithm to check it and directly find consistent parameters.

  The computation of the monoid $\outermonoid{}$ used in \cref{thm:msoParameterLearning} builds upon the following intermediate constructions.
  \begin{enumerate}
    \item {DTA $\dta$ on $\Sigma_1 = \Sigma\times 2^{\{y_1,\dots,y_\ell\}}\times \{?,N,P\}$}
    \item {DFA $\dfa$ on $\Sigma_2 = \Sigma\times 2^{\{y_1,\dots,y_\ell\}}\times \{?,N,P\} \times Q \times \lr$}
    \item {Transition monoid $\innermonoid{}$ for $\dfa$ with final elements $F\!\subseteq\! \innermonoid{}$ and monoid morphism $\freeToInnerMorphism{} \colon \Sigma_2\! \rightarrow \innermonoid{}$}
    \item {Monoid morphism $\freeToInnerMorphism{}' \colon \Sigma_2' \rightarrow \innermonoid{} \text{ with }  \Sigma_2' = \Sigma\times 2^{\{y_1,\dots,y_\ell\}}\times \{?,N,P\} \times \innermonoid{} \times \{L,R\}$}
    \item {Monoid $\outermonoid{} = 2^\innermonoid$ with
      elements $\hat m \subseteq \innermonoid{}$ for every $\hat m \in \outermonoid{}$ and monoid morphism
     $\freeToOuterMorphism{} \colon \Sigma_3 = \Sigma\times \{?,N,P\} \times \outermonoid{} \times \lr \rightarrow \hat M$}   \end{enumerate}

  \subparagraph*{Step 1:}
  The conversion of the formula $\psi$ into a tree automaton $\dta$ is a standard construction (see e.g. \cite{buchi1960weak}).
  It involves extending the tree alphabet $\Sigma$ by unary relations for the free variables of $\psi$ resulting in the alphabet $\Sigma_1 = \Sigma \times \{P,N,?\}\times \parameters$.
  We use a an element from $\braces{P,N,?}$ instead of $2^\braces{P,N}$ since we assumed the training set to be free of contradictions, that is no example appears positive as well as negative.

  \subparagraph*{Step 2:}
  We use the construction from \cref{lem:sameStatesDtaDfa,cor:dfaEquivalentToDta} to convert a DTA with states $Q$ to a DFA with states $Q' = Q\,\dot\cup\,\braces{q_0}$.
  The constructed DFA $\dfa$ is able to simulate a run of $\dta$ on the tree $\tree_2$ over the alphabet $\Sigma_2 = \Sigma_1 \times Q' \times \lr$ where $\tree_2 = \tree$ but the labels are given by $\funcAlphabetExtensionForDfa{}$ from \cref{lem:sameStatesDtaDfa}.
  
  \subparagraph*{Step 3:}
  The monoid $\innermonoid{}$ is the transition monoid of~$\dfa$ computed by the standard construction (see e.g. \cite{tho97a}). 
  $\freeToInnerMorphism{}\colon \Sigma_2 \rightarrow \innermonoid{}$ is the corresponding monoid morphism and $F\subseteq \innermonoid{}$ is the set of accepting monoid elements, both taken from the construction of $\innermonoid{}$.

  \subparagraph*{Step 4:}
  In order to simulate DTAs using DFAs in \cref{lem:sameStatesDtaDfa}, we introduced the component $Q'$ in $\Sigma_2$ which stores the state of the tree automaton on the cut-off subtree at every position.
  In the fourth step we construct the monoid morphism $\freeToInnerMorphism{}'$ that works on $\Sigma_2' = \Sigma\times 2^{\{y_1,\dots,y_\ell\}}\times \{?,N,P\} \times \innermonoid{} \times \lr$; that is, we substitute the component containing for every node $u$ the state $q_u$ of $\dfa$ on the cut-off subtree by a monoid element $m_u\in \innermonoid{}$.
  Hence we modify the tree alphabet to contain a monoid element from $\innermonoid{}$ instead of a state $q\in Q'$.

  From the construction of transition monoids we know that every productive monoid element $m\in\innermonoid$ can be interpreted as a function $f_m: Q' \rightarrow Q'$ where $Q'=Q\cup \braces{q_0}$ is the set of states from $\dfa$.
  Let $f:\innermonoid{} \rightarrow Q'$ be the function defined by $f(m)= f_m(q_0)$.
  This means that for $m\in\innermonoid{}$ and $w\in\Sigma_2^*$ with $\freeToInnerMorphism{}(w) = m$ we have that $f(m) = q$ if and only if $\dfa$ is in state $q$ after reading $w$.
  We now define $\freeToInnerMorphism{}'$ by using $\freeToInnerMorphism{}$ and the function $f$ from above as
  \[
    \freeToInnerMorphism{}'((a_1, m_1, d_1),\dots,(a_n, m_n, d_n)) = \freeToInnerMorphism{}((a_1, f(m_1), d_1),\dots,(a_n, f(m_n), d_n))
  \]
  where $a_1,\dots a_n \in \Sigma_1, m_1,\dots m_n \in \innermonoid{}$ and $d_1,\dots,d_n \in \lr$.
  That is in each label of the input string $w'= (a_1, m_1, d_1),\dots,(a_n, m_n, d_n)$ we substitute $m_i\in \innermonoid{}$ by $f(m_i)\in Q'$ and then apply the monoid morphism $\freeToInnerMorphism{}$.

  With $\innermonoid{}$ and $\freeToInnerMorphism{}'$ we can check for a given $\bar v \in V(\tree)^\ell$ whether $\model{\tree}{\phi}{x}{\bar v}$ is consistent with $\train$. 
  This check would also be possible in linear time using the DTA for $\psi$ directly.
  This more complex construction allows us to compute an index without knowledge of $\train$ and use it to speed up the learning time after getting access to $\train$.

  \subparagraph*{Step 5:}
  In the last step we construct the monoid $\outermonoid{}$ and the monoid morphism $\freeToOuterMorphism{}$ based on $\innermonoid{}$ and $\freeToInnerMorphism{}'$ such that $\freeToOuterMorphism{}$ can be used to check whether there is a consistent set $\bar v \in V(\tree)^\ell$ of parameters.
  We then use $\freeToOuterMorphism{}$ and $\outermonoid{}$ in the actual learning algorithm.
  Let $\outermonoid{} = (2^\innermonoid{}, \cdot_\outermonoid{}, \braces{1_\innermonoid{}})$ be a structure with multiplication 
  \(
    \hat m_1 \cdot_\outermonoid{} \hat m_2 = \{ m_1\cdot_\innermonoid{} m_2 \mid m_1 \in \hat m_1, m_2\in \hat m_2 \}.
  \)
  $\outermonoid{}$ is a monoid as it is closed under multiplication and $\braces{1_\innermonoid{}}\in\outermonoid{}$ is neutral for $\cdot_\outermonoid{}$ because $1_\innermonoid{}\in\innermonoid{}$ is the neutral element of $\innermonoid{}$. 
  Neutrality of $\hat m$ holds since $\hat m \cdot \braces{1_\innermonoid{}} = \braces{m\cdot 1_\innermonoid{} \mid m\in \hat m} = \hat m$ due to the neutrality of $1_\innermonoid{}$ for $\innermonoid{}$.

  We now define a monoid morphism $\freeToOuterMorphism{}\colon \Sigma_3^* \rightarrow \outermonoid{}$ on the basis of $\freeToInnerMorphism{}'$.
  Let $\Sigma_3 := \Sigma \times \{P,N,?\} \times 2^\innermonoid{} \times \{L,R\}$.
  For $(a,\hat m, d) \in \Sigma_3$ with $a \in \Sigma\times \{P,N,?\}$ and $d \in \lr$ we let
  \begin{align*}
    \freeToOuterMorphism{}((a,\hat m,d)) = \!\braces{ \freeToInnerMorphism'((a,\bar y,m,d) ~\middle|~ \bar y \in \parameters, m\!\in \hat m}.
  \end{align*}

  For a position $u\in V(\tree)$ labeled $(a,\hat m,d)\in \Sigma_3$, $\freeToOuterMorphism{}(a,\hat m,d)$ is the set of monoid elements we get by distributing some parameters on $u$, selecting a monoid element $m\in \hat m$ and calling $\freeToInnerMorphism{}'$ on that.
  The existence of $m$ within $\hat m$ implies that the parameters $\monoidParameters{}(m)$ can actually be found in $u$'s cut-off subtree.
  We will use this fact as a divide-and-conquer rule to trace the parameters $\bar y$ to positions $\bar v$ such that the formula with parameters $\bar v$ is consistent with $\train$.
  For a word $w\in \Sigma_3^*$ we split the word into its positions $w_1,\dots,w_n$ and compute the product $\hat m = \prod_{i=1}^n \freeToOuterMorphism{}(w_i)$ of the monoid elements of those.
  This is possible since $\freeToOuterMorphism{}$ is a monoid morphism from $\freemonoid{}^{\Sigma_3}$ to $\outermonoid{}$.
  Note that the interpretation with cut-off subtrees only makes sense for strings $w = \tree_3[P_i]$ that are heavy paths or substrings of those.

  Let $w$ be the string of labels from the heavy path $P_1$ containing the root of $\tree$.
  Since $\freeToOuterMorphism{}$ implicitly tests all possible distributions of parameters in dependent subtrees, we know that if $\freeToOuterMorphism{}(w) \cap F \neq \emptyset$ then there is a consistent parameter setting.
  This holds since every $m\in \freeToOuterMorphism{}(w)$ belongs to at least one distribution $\bar v$ of parameters in $\tree$ and in the other direction any consistent $\bar v$ results in an accepting monoid element $m\in F \subseteq \innermonoid{}$ appearing in $\freeToOuterMorphism{}(w)$.
  Correspondingly there is no consistent parameter setting if $\freeToOuterMorphism{}(w) \cap F = \emptyset$.

  \subparagraph*{Correctness:}
  For the correctness and the actual proof of Lemma~\ref{lem:monoidsInnerAndOuter}, consider first the formula $\psi$.
  Clearly, we have that $\psi$ accepts $P,N$ for some parameter setting $\bar v$ if and only if $\model{\tree}{\phi}{x}{\bar v}$ is consistent with the training set $\train$, given that the unary relation $P$ contains all positive and $N$ all negative examples from $\train$.
  The first conversion from $\psi$ into the DTA $\dta$  is a standard construction such that $\dta$ accepts $\tree_1$ if and only if $\model{\tree}{\phi}{x}{\bar v}$ is consistent with $\train$ where $\tree_1$ is the tree $\tree$ extended by the unary relations $P$ and $N$ and additional unary relations for the parameters $\bar v$.
  For the next intermediate model we use Lemma \ref{lem:sameStatesDtaDfa} to get that the DFA $\dfa$ accepts $\decomp{}(\tree)$ (or rather $\tree_2[P_1]$) if and only if $\dta$ accepts $\tree_1$. 
  In order to run $\dfa$ on $\decomp{}(\tree)$, we extend the alphabet of $\tree_1$ to also include the state of $\dta$ on the cut-off subtree at every position.
  We get from Corollary \ref{cor:dfaEquivalentToDta} that the extended labels for $\tree$ can be computed just-in-time using only $\dfa$ and $\decomp{}(\tree)$.

  The remaining two steps involve monoid structures.
  Since $\innermonoid{}$ is the transition monoid of $\dfa$, there is a set of final elements $F$ such that for every $w\in \Sigma_2^*$, $\freeToInnerMorphism{}(w) \in F$ if and only if $\dfa$ accepts $w$.
  We assume that for every position $i$, the monoid element $m\in\innermonoid{}$ occurring in the label of $P_1[i]$ corresponds to the cut-off subtree at that position.
  Let $\bar v$ be the set of parameters given in the labels of $\decomp{}(\tree)$.
  For the string $w_{P_1}\in \Sigma_2$, $\freeToInnerMorphism{}(w_{P_1}\in F$ if and only if $\model{\tree}{\phi}{x}{\bar v}$ is consistent with $\train$.
  The adaptation $\freeToInnerMorphism{}'$ of $\freeToInnerMorphism{}$ only adds the conversion from a monoid element $m\in \innermonoid{}$ to the corresponding state $q\in Q'$, which hat to be applied outside of $\freeToInnerMorphism{}$ otherwise.
  We use monoid elements in the labels instead of states as this allows us to stay in the monoid domain instead of switching back and forth between states and monoids making the algorithms more readable.

  The remaining transformation is from $\innermonoid{}$ to $\outermonoid{}$ and from $\freeToInnerMorphism{}'$ to $\freeToOuterMorphism{}$.
  By definition, $\freeToOuterMorphism{}$ assigns to a string $w\in \Sigma_3$ the set of monoid elements $\hat m\subseteq \innermonoid{}$, for every distribution of parameters in $w$ and for every choice of $m_i\in \hat m_i$ where $\hat m_i$ is the monoid element in the label of position $i$ speaking about the cut-off subtree at $i$.
  Thus we have by induction that $\freeToOuterMorphism{}$ assigns exactly those monoid elements from $\innermonoid{}$ to $w$ that can be reached for a distribution $\bar v$ of parameters in $\tree$.
  This is exactly the statement of \cref{lem:monoidsInnerAndOuter} as $\innermonoid{}$ checks for consistency of a concrete parameter setting $\bar v$ while $\outermonoid{}$ performs the existential quantification.
\end{inlineproof}

There are more direct ways to check whether there are consistent parameters, for example by existentially quantifying every $y_i$ in $\psi$ and then model checking the resulting formula.
The construction we presented gives us the opportunity to exploit the connection between $\innermonoid{}$ and $\outermonoid{}$ to not only return whether there is a consistent parameter setting $\bar v$, but to actually compute such a parameter setting if there is one.
In order to do this, we categorize the elements from $\innermonoid{}$ from \cref{lem:monoidsInnerAndOuter}.
We know that $\dfa$ only accepts trees in which every parameter, indicated by a unary relation, is assigned exactly once.
Thus, we get that every productive monoid element $m\in \innermonoid{}$, that is we can reach some $m'\in F$ using $m$, has to contain the information which parameters have been read.
Otherwise we would have an accepted tree where a single parameter has been assigned multiple times or not at all and hence we can assign a set of parameters to every productive monoid element.

\begin{lemma}\label{lem:functionMonoidParameters}
  There is a function $\monoidParameters: \innermonoid{} \rightarrow \parameters$ that assigns a set $\bar y \in \parameters$ of parameters to each productive monoid element from $\innermonoid{}$ in a consistent way.
  That is, for a tree $\tree$ and a substring $w$ of a path in $\decomp(\tree)$ with $\freeToInnerMorphism{}'(w)=m$, exactly the parameters $\monoidParameters{}(m)$ occur in the dependent subtree of $w$.
  For every other monoid element $\monoidParameters$ assigns the empty set.
\end{lemma}
\begin{proof}
  For a non-productive $m\in\innermonoid$ there is no unique set of parameters to be assigned.
  This holds as for example it does not matter whether we have read the parameter $y_1$ or the parameter $y_2$ twice, any resulting monoid element will not be accepting in any case.
  We therefore restrict ourselves to productive monoid elements and return the empty set for all others.
  Let $w$ be a string and $\tree_w$ be the dependent subtree of $w$ in $\tree$.
  The set of parameters occurring in $\tree_w$ is defined as the union of parameters occurring in $w$ and the cut-off subtrees for every $u\in w$. 
  Let $\tree$ be a tree containing the information from $\train$ such that there is a consistent parameter setting $\bar v$.
  Let $w$ be a substring from some $P_i \in \decomp{}(\tree)$.
  We define $\monoidParameters{}(m)$ as the set of parameters that occur in the dependent subtree of $w$.

  This does not lead to contradictions as for a second accepted tree $\tree'$ and $w'$ with $\freeToInnerMorphism{}'(w')=\freeToInnerMorphism{}'(w)$ but different parameters in the dependent subtrees of $w$ and $w'$, we could substitute the dependent subtree of $w$ by the one of $w'$ resulting in a tree that needs to be rejected by $\freeToInnerMorphism{}$ as either a parameter is assigned twice or not at all.
  Since $\freeToInnerMorphism{}'(w')=\freeToInnerMorphism{}'(w)$, the monoid morphism $\freeToInnerMorphism{}'$ would map both trees to the same monoid element and thereby either accept both trees or none of them.
\end{proof}

  Note that $\monoidParameters{}(1_\innermonoid{}) = \emptyset$ since $1_\innermonoid{}$ is idempotent ($1_\innermonoid{} = 1_\innermonoid{}1_\innermonoid{}$) and thus if $\monoidParameters{}(1_\innermonoid{})\neq \emptyset$ we would have $\monoidParameters{}(1_\innermonoid{}) \neq \monoidParameters{}(1_\innermonoid{}1_\innermonoid{})$ which does not make sense because $1_\innermonoid{}$ is idempotent.

\subsection{Algorithms}\label{subsec:algorithms}
Using the monoids $\innermonoid{}$ and $\outermonoid{}$ as well as the monoid morphism $\freeToOuterMorphism{}$ from Lemma \ref{lem:monoidsInnerAndOuter}, we can compute a consistent parameter setting $\bar v$ for a given formula $\phi(x\smid \bar y)$ and a training set~$\train$ proving \cref{thm:msoLearning}.
The presented algorithm is split in the following three parts, where the first part is independent of $\train$.
\begin{enumerate}
  \item Indexing: Computation of the auxiliary structure ( $\onot(|\tree|)$)
  \item Updating: Modification of the auxiliary structure to take into account $\train$ ($\onot(|\train|\cdot \log (|\tree|))$)
  \item Tracing: Identification of consistent parameters ($\onot(\ell\cdot \log(|\tree|))$)
\end{enumerate}
The last two parts form the learning phase of the algorithm, while the first part constitutes the precomputation or indexing phase. 
We assume that the underlying formula $\phi$ is fixed and therefore ignore factors depending only on $\phi$.
Recall that those factors might be be non-elementary due to exploding statespaces of the constructed automata making the algorithm a mostly theoretical result.

Formally, we prove \cref{thm:msoParameterLearning} about parameter learning.
Together with a brute-force test of every semantically different formula with quantifier rank and free variables bounded, this proves \cref{thm:msoLearning} which considers the model learning problem and is the main theorem of this paper.

\paragraph*{The indexing algorithm}
The indexing algorithm starts by computing the monoid $\outermonoid$ from $\phi$ as described in \cref{lem:monoidsInnerAndOuter} as well as the heavy path decomposition $\decomp(\tree)$ together with its dependence relation $<_\hp$.
Computing $\outermonoid$ and the corresponding monoid morphism $\freeToOuterMorphism{}$ only depends on $\phi$ and can therefore be achieved in constant time for any fixed formula. 
The computation of $\decomp{}(\tree)$ is linear in $|\tree|$.
The indexing algorithm outputs the set of Simon factorization trees over $\outermonoid{}$ for each $P_i$ from $\decomp{}(\tree)$ computed by the algorithm from \cite{Simon90} as well as $<_\hp$.
Technically we assume $\train = \emptyset$ when extending the labels of $\tree$ to integrate information on the positions of the examples as defined in \cref{lem:monoidsInnerAndOuter}.
We have to use the empty set of examples since the indexing algorithm does not have access to $\train$.
Similar to the case of DFAs simulating a DTA in \cref{cor:dfaEquivalentToDta} we use the dependence relation $<_\hp$ as an order in which the factorization trees are computed.
Essentially a label $a_u\in \Sigma_3$ of a node $u\in V(\tree)$ contains the monoid element $\hat m\in \outermonoid{}$ of the cut-off subtree at $u$ which is available when adhering to the order $<_\hp$.
As the computation of $\outermonoid{}$ only depends on $\phi$, the overall runtime is dominated by the computation of the factorization trees in $\onot(|\tree|\cdot \text{poly}(|\outermonoid{}|))$ by \cref{thm:simon}.

\paragraph*{The update algorithm}

In the update part of the learning algorithm we add the information from the actual training set $\train$ to the factorization trees computed in the indexing part of the algorithm.
It is easy to see that recomputing the factorization trees of all modified heavy paths may take time $\onot(|\tree|)$ as a single heavy path may be linear in the size of $\tree$.
In the presented algorithm the updates are performed bottom-up, that is we change the labels in the leaves of the factorization tree, update those and then propagate this information towards the root of $\tree$ resulting in further updates in factorization trees.
The algorithm works in two stages: an outer stage that collects all updates for each heavy path and an inner stage that actually performs the update.

The outer stage orchestrates the update process by computing the set of updates $\labelChanges_{P_i}$ for every $P_i\in \decomp{}(\tree)$.
In the inner stage we use an update algorithm that given a set $\labelChanges{}$ of label changes and a factorization tree $F$ returns a factorization tree $F'$ about twice as high independent of $\labelChanges{}$.
In order to maintain a bounded height of the updated factorization trees, we bound the number of update steps for each factorization tree.
The set $\labelChanges_{P_i}$ consists of updates for every $u\in P_i \cap \train$ as well as (possible) updates from paths $P_j$ with $P_j <_\hp P_i$ due to previous updates.
Both kind of updates change the label of a single node from $\tree$ and thus can be treated in the same way.
An update at $u$ due to an example $(u,+) \in \train$ modifies the component $\braces{P,N,?}$ of $u$'s label.
Updates in the cut-off subtree of $u$ may induce a change in the monoid element $\hat m$ in the label of $u$.
This holds as $\hat m$ depends on the whole cut-off subtree at $u$.
If we update the factorization tree for $P_i\in \decomp{}(\tree)$ after the factorization trees for all paths $P_j <_\hp P_i$ have been updated, every factorization tree is updated at most once.
There are no subsequent or late updates on $F_{P_i}$ as label changes from $\train$ are known from the start of the update algorithm and all other label changes are due to updates in the dependent subtree of $P_i$ containing only paths $P_j <_\hp P_i$. 

The inner algorithm, which actually performs the update, is taken from \cref{lem:updateFactorizationTree}.
By \cref{lem:logarithmicHeavyPathsToTheRoot} every path from a node to the root touches at most logarithmically many heavy paths resulting in a total runtime of the update algorithm of $\onot(|\train|\log |\tree|)$ viewing $|\outermonoid{}|$ as constant as it only depends on $\phi$.

\paragraph*{The tracing algorithm}
The tracing algorithm gets the updated set of factorization trees $\mathcal F = F_{P_1},\dots,F_{P_n}$ and computes a set of parameters $\bar v$ such that $\model{\tree}{\phi}{x}{\bar y}$ is consistent with $\train$.
We again divide the algorithm into an inner and an outer algorithm where the inner algorithm traces the parameters similar to~\cite{DBLP:conf/alt/GroheLR17} within a single factorization tree and the outer algorithm orchestrates the search over $\mathcal F$.
The complete algorithm, including the inner and outer part of the tracing process, is given in pseudocode in \cref{alg:trace}.
It uses two stacks to track the open tasks of the inner and outer algorithm respectively.
The stack for the outer algorithm contains factorization trees and target monoids, such as $(F_{P_i},m_i)$, while the stack for the inner algorithm additionally stores the current position within the factorization tree.

Let $P_1$ contain the root of $\tree$ and thus be maximal according to $<_\hp$.
Let $\hat m_r\in\outermonoid{}$ be the monoid element reached in the root of $F_{P_1}$.
The tracing starts by choosing the (initial) local target $m_r\in \hat m_r \cap F$ for the root of $F_{P_1}$ where $F$ is the set of final monoid elements from \cref{lem:monoidsInnerAndOuter}. 
If $\hat m_r \cap F = \emptyset$, then there is no assignment of parameters to positions of $\tree$ that is consistent with $\train$ and the algorithm stops.
Note that the choice of $m_r$ is arbitrary, as every accepting monoid element corresponds to an assignment $\bar v$ of the parameters $\bar y$ such that $\model{\tree}{\phi}{x}{\bar v}$ is consistent with $\train$.
Thus we have $\monoidParameters{}(m_r) = \braces{y_1,\dots,y_\ell}$ for every possible choice of $m_r$ such that all parameters will be traced to a node of $\tree$. 

For a given factorization tree $F_{P_i}$, the algorithm traces the set of parameters in $F_{P_i}$ based on \cref{lem:functionMonoidParameters} to leaves of $F_{P_i}$ and then continues with the largest (according to $<_\hp$) remaining heavy path that contains parameters according to the previous tracing steps.
Let $u$ be position in $F_{P_i}$ with local target $m_u$.
Then we select monoid elements $m_1,m_2$ from the children $u_1,u_2$ of $u$ such that $m_1m_2 = m_u$ using a brute-force test.
The tracing continues this way for every $u_i$ where $\monoidParameters{}(m_i) \neq \emptyset$ with $m_i$ as local target at $u_i$ until the leaves of $F_{P_i}$ are reached and outputs pairs $(u,m_u)\in V(\tree)\times \innermonoid{}$ with $\monoidParameters{}(m_u)\neq \emptyset$.
An open task $(F,u,m_u)$ for the inner algorithm consists of the factorization tree $F$, the current position $u\in F$ and the (local) target monoid $m_u$ for the node $u$.
It then checks whether there are monoid elements $m_1,m_2$ at the children $u_1,u_2$ of $u$ such that $m_1m_2 = m_u$ and pushes $(F,u_i,m_i)$ for $i\in \braces{1,2}$ if $\monoidParameters(m_i)\neq \emptyset$.
When $u$ is a leaf in $F$ the algorithm outputs the tuple $(F,u,m_u)$ which is then interpreted by the outer algorithm.
In \cref{alg:trace} this tracing within a single factorization tree is done in lines \ref{alg:trace:line:startInner} to \ref{alg:trace:line:endInner}.

For a tuple $(F,u,m_u)$ from the inner algorithm where $u$ has the label $(a,\train_u,\hat m_{u'},D)\in \Sigma_3$, the outer algorithm partitions the parameters  $\monoidParameters{}(m_u)$ into sets $K_1,K_2\subseteq \monoidParameters{}(m_u)$.
The parameters $K_1$ are the ones which are placed at $u$ while those from $K_2$ are further traced in the cut-off heavy path $P_j$ at $u$.
This is implicitly done by selecting a monoid element $m_{u'}$ for the cut-off heavy path which then implies that $K_2 = \monoidParameters{}(m_{u'})$ and $K_1 = \monoidParameters{}(m_u\setminus K_2$.
For every $m_{u'} \in \hat m_{u'}$ the algorithm checks whether $\freeToInnerMorphism{}((a,K_1,\train_u,m,D) = m_u$.
This means that the algorithm brute-force checks for every $m_{u'}\in \hat m_u'$ that is reachable in the cut-off subtree by a distribution of the parameters whether it yields $m_u$ at $u$.
The brute-force test will always be successful since the monoid element $\hat m \in \outermonoid{}$ at $u$ contains exactly those monoid elements reachable via a combination of parameters at $u$ and some $m_{u'} \in \hat m_{u'}$ for the cut-off subtree.
The monoid element $m_{u'}$ found this way is then, together with the root the factorization tree $F_{P_j}$ for which $m_{u'}$ was selected, given to the inner tracing algorithm.
Additionally the parameters from $K_1$ are fixed at $u$.

The DFA underlying the monoid structures only accept trees where every parameter has been assigned to exactly one position, and thus we have for every $m\in F$ that $\monoidParameters{}(m) = \{y_1,\dots,y_\ell\}$.
Since the tracing algorithm starts with some $m\in F$ in the root of $F_{P_1}$, we know that by continuing the above steps the algorithm computes an assignment $\bar u$ for the parameters $\bar y$.

\newcommand\mycommfont[1]{\footnotesize\ttfamily\textcolor{blue}{#1}}
\SetCommentSty{mycommfont}

\begin{algorithm}[t]
\LinesNumbered
\DontPrintSemicolon

\KwIn{Tree $\tree$, Simon factorization trees $F_{P_1},\dots F_{P_n}$ of the heavy path decomposition $P_1,\dots,P_n$ of $\tree$ and the partial order $<_\hp{}$}
\KwOut{For each of the parameters $y_1,\dots,y_\ell$ a position $u_1,\dots,u_\ell$ from $\tree$ such that $\model{\phi}{x}{\bar u}{\tree}$ is consistent with $\train$}

  $P^* \gets \max_{<_\hp}(\braces{P_1,\dots,P_n})$

  $\bar m \gets \text{label}(\text{root}(F_{P^*}))$

  \If {$\bar m \cap F = \emptyset$}{
    abort 
    \tcp*[f]{there are no consistent parameters}
  }

  $m \in \bar m \cap F$

  outer.push($(F_{P^*},m)$)

  \Repeat{$\text{outer.isempty}()$}{

  $(F,m) \gets \text{outer.pop}()$

  inner.push($(F,\text{root}(F),m)$)

  \newcommand{\current}{u}
  \Repeat{$\text{inner.isempty}()$}{
    $(F,\current,m) \gets \text{inner.pop}()$\label{alg:trace:line:startInner}

    \uIf {$\neg \text{leaf}_{F}(\current)$}{
      $u_1 \gets$ leftChild($F,\current$)

      $u_2 \gets$ rightChild($F,\current$)

      $\bar m_1 \gets \text{label}(u_1)$ 

      $\bar m_2 \gets \text{label}(u_2)$

      $m_1,m_2 \in \braces{m_1,m_2 \mid m_1m_2=m, m_1\in \bar m_1, m_2\in \bar m_2}$
      \tcp*[f]{exhaustive search}

      \If{$\monoidParameters{}(m_1)\neq \emptyset$}{inner.push($F,u_1,m_1$)}

      \If{$\monoidParameters{}(m_2)\neq \emptyset$}{inner.push($F,u_2,m_2$)}\label{alg:trace:line:endInner}

    }
    \Else{
      \If{$\text{params}(m)\neq \emptyset$}{\label{alg:trace:line:startLeaf}
        find $k\in \text{params}(m)$ and $m' \in \bar m' = \project{}(\text{label}(\current))$ with $\freeToOuterMorphism{}(a,k,m',\dots) = m$ 
        \tcp*[f]{exhaustive search}\label{alg:trace:line:localParams}

        \If{$\monoidParameters{}(m')\neq \emptyset$}{
          outer.push($\text{cutoff}(\current),m'$)
        }

        \ForEach{$y\in k$}{
          $ys[y] \gets \current$
        }
      }\label{alg:trace:line:endLeaf}
    }
  }
  }
  \Return $ys$

\caption{Computing a consistent parameter setting for trees.}\label{alg:search}\label{alg:traversalBinaryTreeUnaryRelation}\label{alg:trace}
\end{algorithm}

For each parameter $y_i$, the tracing algorithm works on at most $\log |\tree|$ heavy paths by \cref{lem:logarithmicHeavyPathsToTheRoot} and uses constant time within each factorization tree.
Correspondingly, by using the update algorithm from \cite{DBLP:conf/alt/GroheLR17} and \cref{lem:logarithmicHeavyPathsToTheRoot} we get a runtime of $\onot(|\train|\cdot \log |\tree|)$ for the updating algorithm.
In total the time needed to find a consistent parameter setting $\bar v\in V(\tree)^\ell$ for $\phi$ is an indexing time in $\onot{}(|\tree|)$ and a search time in $\onot(\log |\tree|\cdot (|\train|+\ell))$.

Observe that for a monoid element $\hat m$ of a node $u$ every $m\in \hat m$ can be reached.
This holds by induction as for the leaves we add those $m$ to $\hat m$ which can be reached by assigning any possible subset $K\in \parameters{}$ of the parameters to such a leaf.
For inner nodes, we know that this holds for the cut-off subtree by induction. 
Then $\freeToOuterMorphism{}$ creates $\hat m$ by computing the set of all monoid elements $m\in\innermonoid{}$ reachable by combinations of monoid elements $m'\in \hat m'$ for the cut-off subtree and parameters assigned to that position.

We know that if there is a consistent parameter setting $\bar v$, then by \cref{lem:monoidsInnerAndOuter} there exists $\hat m$ with $\hat m\cap F \neq \emptyset$ in the root of $F_{P_1}$  .
The parameter configuration $\bar v$ found by the algorithm is consistent with $\train$ as it is computed in a way that $\freeToInnerMorphism{}'(\tree_2[P_1]) = m_r$ with $\in F$.
The search for parameters is always successful since for every $m\in \hat m$ at some node $u$ there is at least one way to distribute a subset of the parameters $\bar y$ in the cut-off subtree of $u$ such that $m$ represents the cut-off subtree, that is $\freeToInnerMorphism{}$ returns $m$ on it.
The correctness of the inner tracing algorithm follows directly from \cite{DBLP:conf/alt/GroheLR17}.
For the outer part of the algorithm it brute-force checks locally every possible distribution of parameters among a node $u$ and its cut-off subtree starting from an accepting monoid element $m\in F$.
The positions $\bar v$ for the parameters $\bar y$ found this way are consistent since $m\in F$ and integrating $\bar v$ in the labels of $\tree$ resulting in $\tree'$ ensures that $\freeToInnerMorphism{}'(\decomp{}(\tree')) = m \in F$ which is accepting.
Whenever there is a consistent $\bar v$, then for the monoid element $\hat m_u$ in the root of $F_{P_1}$ there is some $m\in \hat m_u\cap F$ by the construction of $\psi$ and $\outermonoid{}$ such that a solution will be found by the algorithm.

In the constructed set of factorization trees $F_{P_1},\dots,F_{P_n}$ we assume that for every node $u$ the monoid element $m$ in the label of $u$ is exactly the monoid element returned by $\freeToOuterMorphism{}$ on the cut-off subtree.
This is guaranteed by the indexing and updating algorithm since they work bottom-up and simply copy the computed monoid element $m$ of $u$'s cut-off subtree into its label.

\section{Online learning of MSO formulas}

For Theorem \ref{thm:msoLearning} we assumed that after an indexing phase the complete training set $\train$ is known and the learning task is to find a hypothesis consistent with $\train$. 
We now lift this result to an online setting where we again have a linear indexing phase and then on input of new batch of examples $\train_i$ the algorithm updates its hypothesis $H$ such that $H_i$ is consistent with all examples $\train = \bigcup_i \train_i$ it has seen so far.
This online setting allows examples to arrive over time which is natural for many tasks with human interaction.
Note that the algorithm can also handle label updates of the nodes as both types of updates induce a label change in the tree $\tree$.
Label changes are a common type of update in a database setting since updating the attributes of an already present entity can be modeled by an update of the entity's label.

An \emph{online learning algorithm} is an algorithm that takes as input a background structure~$\tree$, an index $I(\tree)$ and a sequence $\train = (u_1,c_1),(u_2,c_2),\dots$ of training examples and outputs for every $i\leq |\train|$ a hypothesis $H_i = \model{\tree}{\phi}{x}{\bar v}$ consistent with $\train_i = \{(u_1,c_1),\dots,(x_i,c_i)\}$.

\begin{theorem}\label{thm:onlineLearning}
  Let $q,\ell\in \NN$. There is
  an indexing algorithm $A$ that, given a tree $\tree$ computes an index $A(\tree)$ and
  an online learning algorithm $B$ that, given $\tree, I(\tree)$, and an $MSO[q,\ell+1]$-realizable sequence $\train = (u_1, c_1), \dots$ for $\tree$, maintains a consistent hypothesis $H_i = \model{\tree}{\phi}{x}{\bar v}$ for every $i\in\NN$
  such that $A$ runs in time $\onot(|\tree|)$ and $B$ runs in time $\onot{}(\log^2 (|\tree|))$ per update.
\end{theorem}

This can be achieved by substituting Simon factorization trees by the conceptually simpler binary factorization trees in the construction.
This implies two main differences.
First, we can update labels arbitrarily often without changing the structure of the factorization tree (which does not hold for Simon factorization trees due to the idempotent elements).
Second, the height of each factorization tree is logarithmic in its length, \ie at most logarithmic in~$\tree$. 
Therefore it takes logarithmic time to update each path and since a single update may involve updating logarithmically many paths this results in a runtime of $\onot{}(|\train_i| \log^2 (|\tree|))$ per update.

\section{Conclusion}
We considered the setting of learning quantifier-free and MSO formulas on trees.
All learning algorithms provided in this paper search for consistent hypotheses, thus they can be turned into PAC learning algorithms by providing a large enough training set. 

We assumed the background structures to be huge and therefore have been researching sublinear algorithms which access to the background structures through the local access oracles.
The first result is that even for quantifier free formulas there is no sublinear learning algorithm.
However, there is a sublinear learning algorithm when given access to the largest common ancestor of two nodes.
Our main result is a learning algorithm for unary MSO formulas which uses a linear indexing phase to build up an auxiliary structure (the index) and admits a logarithmic learning time with local access to that index.

Further research questions might include lifting the result to higher dimensions where examples consist of pairs or tuples of nodes instead of single positions in the tree.
Another direction of research could be to extend our results for tree-like structures.
For structures of bounded tree-width the approach could use a similar structure as the one from \cite{courcelle1990monadic}.
A slightly different research question would be to look for approximate solutions where only a certain (relative) amount of examples needs to be consistent.
Such approaches could also deal with faulty examples, which occur quite regularly in practice.

\bibliographystyle{plain}
\bibliography{quellen}

\begin{thebibliography}{10}

\bibitem{aboangpap+13}
A.~Abouzied, D.~Angluin, C.H. Papadimitriou, J.M. Hellerstein, and
  A.~Silberschatz.
\newblock Learning and verifying quantified boolean queries by example.
\newblock In R.~Hull and W.~Fan, editors, {\em Proceedings of the 32nd {ACM}
  {SIGMOD-SIGACT-SIGART} Symposium on Principles of Database Systems}, pages
  49--60, 2013.

\bibitem{Angluin78}
D.~Angluin.
\newblock On the complexity of minimum inference of regular sets.
\newblock {\em Information and Control}, 39(3):337--350, 1978.

\bibitem{Angluin87}
D.~Angluin.
\newblock Learning regular sets from queries and counterexamples.
\newblock {\em Information and Computation}, 75(2):87--106, 1987.

\bibitem{Angluin90}
D.~Angluin.
\newblock Negative results for equivalence queries.
\newblock {\em Machine Learning}, 5:121--150, 1990.

\bibitem{BalminPV04}
A.~Balmin, Y.~Papakonstantinou, and V.~Vianu.
\newblock Incremental validation of {XML} documents.
\newblock {\em {ACM} Trans. Database Syst.}, 29(4):710--751, 2004.

\bibitem{bluehrhau+89}
A.~Blumer, A.~Ehrenfeucht, D.~Haussler, and M.K. Warmuth.
\newblock Learnability and the {V}apnik-{C}hervonenkis dimension.
\newblock {\em Journal of the ACM}, 36:929--965, 1989.

\bibitem{Bojanczyk12}
M.~Boja{\'{n}}czyk.
\newblock Algorithms for regular languages that use algebra.
\newblock {\em {SIGMOD} Record}, 41(2):5--14, 2012.

\bibitem{bonciusta16}
A.~Bonifati, R.~Ciucanu, and S.~Staworko.
\newblock Learning join queries from user examples.
\newblock {\em {ACM} Trans. Database Syst.}, 40(4):24:1--24:38, 2016.

\bibitem{buchi1960weak}
J~Richard B{\"u}chi.
\newblock Weak second-order arithmetic and finite automata.
\newblock {\em Mathematical Logic Quarterly}, 6(1-6):66--92, 1960.

\bibitem{cohpag95}
W.W. Cohen and C.D. Page.
\newblock Polynomial learnability and inductive logic programming: Methods and
  results.
\newblock {\em New generation Computing}, 13:369--404, 1995.

\bibitem{Colcombet11}
T.~Colcombet.
\newblock Green's relations and their use in automata theory.
\newblock In {\em Language and Automata Theory and Applications - 5th
  International Conference, {LATA} 2011, Tarragona, Spain, May 26-31, 2011.
  Proceedings}, volume 6638 of {\em Lecture Notes in Computer Science}, pages
  1--21. Springer, 2011.

\bibitem{courcelle1990monadic}
B.~Courcelle.
\newblock The monadic second-order logic of graphs. i. recognizable sets of
  finite graphs.
\newblock {\em Information and computation}, 85(1):12--75, 1990.

\bibitem{drewes2003learning}
F.~Drewes and J.~H{\"o}gberg.
\newblock Learning a regular tree language from a teacher.
\newblock In {\em Developments in Language Theory}, pages 279--291. Springer,
  2003.

\bibitem{garneimadrot16}
P.~Garg, D.~Neider, P.~Madhusudan, and D.~Roth.
\newblock Learning invariants using decision trees and implication
  counterexamples.
\newblock In {\em Proceedings of the 43rd Annual {ACM} {SIGPLAN-SIGACT}
  Symposium on Principles of Programming Languages}, pages 499--512, 2016.

\bibitem{Gold78}
E.M. Gold.
\newblock Complexity of automaton identification from given data.
\newblock {\em Information and Control}, 37(3):302--320, 1978.

\bibitem{DBLP:conf/alt/GroheLR17}
M.~Grohe, C.~L{\"{o}}ding, and M.~Ritzert.
\newblock Learning mso-definable hypotheses on strings.
\newblock In {\em International Conference on Algorithmic Learning Theory,
  {ALT} 2017, 15-17 October 2017, Kyoto University, Kyoto, Japan}, pages
  434--451, 2017.

\bibitem{grorit17}
M.~Grohe and M.~Ritzert.
\newblock Learning first-order definable concepts over structures of small
  degree.
\newblock In {\em Proceedings of the 32nd ACM-IEEE Symposium on Logic in
  Computer Science}, 2017.

\bibitem{grohe2004learnability}
M.~Grohe and G.~Tur{\'a}n.
\newblock Learnability and definability in trees and similar structures.
\newblock {\em Theory of Computing Systems}, 37(1):193--220, 2004.

\bibitem{harel1984fast}
Dov Harel and Robert~Endre Tarjan.
\newblock Fast algorithms for finding nearest common ancestors.
\newblock {\em siam Journal on Computing}, 13(2):338--355, 1984.

\bibitem{jorkai16}
C.~Jordan and L.~Kaiser.
\newblock Machine learning with guarantees using descriptive complexity and smt
  solvers.
\newblock {\em ArXiv (CoRR)}, arXiv:1609.02664 [cs.LG], 2016.

\bibitem{KearnsValiant94}
M.J. Kearns and L.G. Valiant.
\newblock Cryptographic limitations on learning boolean formulae and finite
  automata.
\newblock {\em Journal of the ACM}, 41(1):67--95, 1994.

\bibitem{kiedze94}
J.-U. Kietz and S.~Dzeroski.
\newblock Inductive logic programming and learnability.
\newblock {\em {SIGART} Bulletin}, 5(1):22--32, 1994.

\bibitem{lodmadnei16}
C.~L{\"{o}}ding, P.~Madhusudan, and D.~Neider.
\newblock Abstract learning frameworks for synthesis.
\newblock In M.~Chechik and J.{-}F. Raskin, editors, {\em Proceedings of the
  22nd International Conference on Tools and Algorithms for the Construction
  and Analysis of Systems}, volume 9636 of {\em Lecture Notes in Computer
  Science}, pages 167--185. Springer Verlag, 2016.

\bibitem{mug91}
S.~Muggleton.
\newblock Inductive logic programming.
\newblock {\em New Generation Computing}, 8(4):295--318, 1991.

\bibitem{mug92}
S.H. Muggleton, editor.
\newblock {\em Inductive Logic Programming}.
\newblock Academic Press, 1992.

\bibitem{mugder94}
S.H. Muggleton and L.~De Raedt.
\newblock Inductive logic programming: Theory and methods.
\newblock {\em The Journal of Logic Programming}, 19-20:629--679, 1994.

\bibitem{GarciaO92}
J.~Oncina and P.~Garc\'{i}a.
\newblock Identifying regular languages in polynomial time.
\newblock In {\em Proceedings of the International Workshop on Structural and
  Syntactic Pattern Recognition}, volume~5 of {\em Machine Perception and
  Artificial Intelligence}, pages 99---108. World Scientific, 1992.

\bibitem{PittW93}
L.~Pitt and M.K. Warmuth.
\newblock The minimum consistent {DFA} problem cannot be approximated within
  any polynomial.
\newblock {\em Journal of the ACM}, 40(1):95--142, 1993.

\bibitem{RabinS59}
M.O. Rabin and D.Scott.
\newblock Finite automata and their decision problems.
\newblock {\em IBM Journal of Research and Development}, 3:114--125, 1959.

\bibitem{RivestS93}
R.L. Rivest and R.E. Schapire.
\newblock Inference of finite automata using homing sequences.
\newblock In {\em Machine Learning: From Theory to Applications}, volume 661 of
  {\em Lecture Notes in Computer Science}, pages 51--73. Springer, 1993.

\bibitem{Simon90}
I.~Simon.
\newblock Factorization forests of finite height.
\newblock {\em Theoretical Computer Science}, 72(1):65--94, 1990.

\bibitem{staworko2012learning}
S{\l}awek Staworko and Piotr Wieczorek.
\newblock Learning twig and path queries.
\newblock In {\em Proceedings of the 15th International Conference on Database
  Theory}, pages 140--154. ACM, 2012.

\bibitem{tho97a}
W.~Thomas.
\newblock Languages, automata, and logic.
\newblock In G.~Rozenberg and A.~Salomaa, editors, {\em Handbook of Formal
  Languages}, volume~3, pages 389--456. Springer-Verlag, 1997.

\bibitem{valiant1984theory}
L.G. Valiant.
\newblock A theory of the learnable.
\newblock {\em Communications of the ACM}, 27(11):1134--1142, 1984.

\bibitem{vapnik2015uniform}
V.~Vapnik and A.~Chervonenkis.
\newblock On the uniform convergence of relative frequencies of events to their
  probabilities.
\newblock {\em Theory of Probability and its Applications}, 16:264--280, 1971.

\bibitem{weiss2017reverse}
Y.~Weiss and S.~Cohen.
\newblock Reverse engineering spj-queries from examples.
\newblock In {\em Proceedings of the 36th ACM SIGMOD-SIGACT-SIGAI Symposium on
  Principles of Database Systems}, pages 151--166. ACM, 2017.

\end{thebibliography}

\end{document}